\documentclass[11pt, oneside, reqno]{amsart}  

\title{Algebra, Geometry and Topology of ERK Kinetics}   %note \\[1ex] is a line break in the title

\author{Lewis Marsh \and Emilie Dufresne \and Helen M. Byrne  \and Heather A. Harrington }

\usepackage[applemac]{inputenc}
\usepackage[fleqn]{mathtools}
\usepackage{regexpatch}

\usepackage{amssymb,amsthm,mathrsfs}
\usepackage{enumerate}
\usepackage{dsfont}
\usepackage{graphicx}
\usepackage{pdfpages}
\usepackage[caption=false]{subfig}
\usepackage[ruled,vlined,linesnumbered,displayblockmarkers]{algorithm2e}
\usepackage{tikz}
\usepackage{bbm}
\usetikzlibrary{positioning}
\usepackage{tikz-cd}
\usepackage{float}
\usepackage{array}
\usepackage[left=2.5cm,right=2.5cm, top=2.5cm, bottom=2.5cm]{geometry}
\usepackage{xcolor}

\usepackage[british]{babel}
\usepackage[babel]{csquotes}

\usepackage[backend=bibtex, date=long, giveninits=true, terseinits=true, maxnames=4]{biblatex}
\bibliography{report.bib}
\usepackage{nameref}
\usepackage{enumitem}
\renewcommand{\refname}{ALGEBRA, GEOMETRY AND TOPOLOGY OF ERK KINETICS}

\DeclareNameAlias{sortname}{family-given}
\DeclareNameAlias{author}{family-given}

\newcommand{\etal}{et al. }

\newtheorem{lemma}{Lemma}
\newtheorem{theorem}[lemma]{Theorem}

\newtheorem{proposition}[lemma]{Proposition}
\newtheorem*{result}{Biological Result}
\newtheorem{problem}[lemma]{Open Problem}

\theoremstyle{definition}
\newtheorem{definition}[lemma]{Definition}
\theoremstyle{remark}

\newtheorem*{remark*}{Remark}

\DeclareMathOperator{\RR}{\mathbb R}
\DeclareMathOperator{\NN}{\mathbb N}
\DeclareMathOperator{\ZZ}{\mathbb Z}
\DeclareMathOperator{\FF}{\mathbb F}
\DeclareMathOperator{\CC}{\mathbb C}
\DeclareMathOperator{\KK}{\mathcal K}

\usepackage[pdftex, hidelinks]{hyperref}
            
\usepackage[symbols,nogroupskip,sort=none]{glossaries-extra}

\begin{document}

\maketitle                  % create a title page from the preamble info

\begin{abstract}
The MEK/ERK signalling pathway is involved in cell division, cell specialisation, survival and cell death \cite{shaul2007mek}. Here we study a polynomial dynamical system describing the dynamics of MEK/ERK proposed by Yeung \etal \cite{Yeung2020} with their experimental setup, data and known biological information. The experimental dataset is a time-course of ERK measurements in different phosphorylation states following activation of either wild-type MEK or MEK mutations associated with cancer or developmental defects. We demonstrate how methods from computational algebraic geometry, differential algebra, Bayesian statistics and computational algebraic topology can inform the model reduction, identification and parameter inference of MEK variants, respectively. 
Throughout, we show how this algebraic viewpoint offers a rigorous and systematic analysis of such models.

\end{abstract}

\color{black}

\section{Introduction}

In systems biology, dynamics play a crucial role in cellular decision making (e.g., whether a cell responds appropriately to a particular signal) \cite{klipp2013systems,voit2017first}. Molecular interactions can be modelled as systems of chemical reactions with a choice of kinetics, such as the law of mass action, which assumes that the rate at which a chemical reaction proceeds is proportional to the product of the concentrations of its reactants. From a finite set of reactions, the mass-action modelling assumption gives rise to a system of polynomial ordinary differential equations (ODEs), which are sums of monomials in which each term includes concentrations of molecular species 
as variables and coefficients  
as rates of reaction. Chemical reaction network theory (CRNT) is a mathematical field developed by Horn and Jackson, and independently by Bykov, Gorban, Volpert and Yablonsky, for analysing such reactions, and the mathematical techniques employed extend beyond dynamical systems theory to include algebraic geometry, differential algebra, algebraic statistics, and discrete mathematics \cite{dickenstein2016biochemical}. 

CRNT often focuses on steady-state analysis through the lens of computational and real algebraic geometry, asking questions about the capability or preclusion of multiple positive real steady-states (i.e., multistationarity) or more complex dynamics, often without requiring specialised parameter values \cite{banaji2009graph,craciun2005multiple,millan2012chemical,angeli2009tutorial,wang2008number,feliu2012preclusion,muller2016sign,conradi2019multistationarity}). Multi-site protein phosphorylation systems, such as the ERK/MEK signalling pathway can be translated into such chemical reactions and their multistationarity, corresponding to different biological cellular decisions, has attracted much attention \cite{thomson2009unlimited,gunawardena2007distributivity,aoki2011processive,takahashi2010spatio,markevich2004signaling}. Algebraic analyses and invariants of multi-site phosphorylation have revealed geometric information of steady-state varieties, informed experimental design, and enabled model the comparison using steady-state data \cite{manrai2008geometry,thomson2009rational,harrington2012parameter,gross2016algebraic,maclean2015parameter}. However, such systems have also been shown to exhibit nontrivial transient dynamics and oscillations \cite{conradi2019dynamics,qiao2007bistability}. In recent years, the fields of systems biology and CRNT have extended the repertoire of techniques to assert other dynamics \cite{banaji2020building,conradi2019dynamics,domijan2009bistability,mincheva2007graph,kay2017role,errami2015detection,angeli2013combinatorial},
reduce models systematically \cite{pantea,reduction,feliu2019quasisteady,sweeney2017conditions,boulier2011model,hubert2013scaling},
and assess identifiability \cite{ljung1994global,ollivier1990probleme,meshkat2009algorithm,Hong2019,DAISY}. Furthermore, combinatorial structures, such as simplicial complexes, and techniques from computational algebraic topology have enabled comparison of chemical reaction network models and their parameters \cite{vittadello2020model,nardini2020topological}. 

A previous algebraic systems biology case study \cite{gross2016algebraic} analysed a chemical reaction network model at steady-state, 
by studying the steady-state ideal, chamber complex, and algebraic matroids of the model. Here we present a sequel of such analysis to study the dynamics of chemical reaction networks with \textit{time-course data}, which relies on studying the QSS variety (Section 3), the model prediction map (Section 4) and the topology of a parameter inference (Section 5).

We perform a detailed mathematical analysis of recently published models and experimental data \cite{Yeung2020}. The \textit{Full ERK model} describes dual phosphorylation of ERK by MEK, two molecular species whose activation regulates cell division, cell specialisation, survival and cell death \cite{shaul2007mek}.
The dynamics of the six ERK/MEK molecular species $x\in\RR^{n=6}$ in the \textit{Full ERK model} are governed by a polynomial dynamical system $\dot{x}(t) = f(x(t),\theta)$, where $\theta\in\RR^{m=6}$ is the vector of parameters and there are two conservation relations between the species. The Full ERK model is presented in Section 2.
Analysing the kinetic parameters of a model
depends on the available data. The accompanying time-course experimental observations include measurements of ERK in 3 different states, at 7 time points following activation by its activated enzyme kinase MEK, which is either wild-type (WT) or mutated MEK. Mutations of MEK are known to be involved in human cancer and embryonic developmental defects; therefore understanding their kinetics and differences between wild-type and 4 mutants (e.g., Y130C, F53S, E203K or SSDD) may increase fundamental biological understanding of the pathway and contribute to the development of potential therapies. The experimental data and relevant biological information are presented in Section~\ref{sec:model}.

Using algebraic approaches first presented by Goeke, Walcher and Zerz in \cite{reduction}, we decrease the number of variables and parameters in the FUll ERK model. We derive two model reductions: the \emph{Rational ERK model} and the \emph{Linear ERK model}.  We show, with known biological information (see Section 2), that the reduction to the {Linear ERK model} by Yeung \etal \cite{Yeung2020} is mathematically sound. We note that the Rational ERK model was not analysed in \cite{Yeung2020}, although it can be derived from the Full ERK model using singular perturbation methods. A natural question is whether a quasi-steady-state approximation is justified given the experimental setup, which equates to solving an algebraic problem \cite{reduction}. We identify algebraic varieties $V_\theta$ that are (analytic) invariant sets of the ODE system and characterise neighbourhoods in parameter space for which the ODE solutions remain close to these varieties.  This systematic analysis allows us to simplify the model equations such that the dynamics of both reduced models are good approximations to the Full ERK model.  Algebraic model reduction and derivation of the reduced ERK models are given in Section~\ref{sec:algred}.

%%%%%%%%%

Before estimating the parameters of a model from observations, one must determine its identifiability.  
Identifiability is concerned with asking whether it is possible to recover values of the model parameters given data. A model is \textit{structurally identifiable} if parameter recovery is possible with perfect data. Mathematically, this task is equivalent to asking whether the model prediction map
is injective. The model prediction map, defined precisely in Subsection~\ref{subsec:SI}, is a map that takes a parameter to the corresponding predicted noise-free data point(s) \cite{Dufresne2018}. Real data is often noisy; testing whether parameter recovery is possible with imperfect data is the problem of \textit{practical identifiability} \cite{raue2009structural,Dufresne2018}. 
Mathematically, measurement noise induces a probability distribution in data space. Assuming that the model prediction map is injective (at least generically), practical identifiability can be defined in terms of the boundedness (with respect to a reference metric in parameter space) of the confidence regions of a likelihood test. Under our assumptions, this translates to asking whether the preimages of small bounded regions in data space are bounded in parameter space.
We prove the following:
\begin{theorem} \label{thm:LinearERKIdent}
The Linear ERK/MEK model, with the given experimental setup (number of species, number of replicates, number of measurement time-points and initial conditions), is structurally and practically identifiable.
\end{theorem}

We provide a definition of practical identifiability that improves
a previous definition \cite{Dufresne2018}, and which is an alternative to that of Raue et al. \cite{raue2009structural}. We also propose a computable algorithm for practical identifiability. We prove Theorem 1 in Section~\ref{sec:ident}. 

We use the differential algebra method to show that the Full ERK model and the Rational ERK model are generically structurally identifiable. This result is guaranteed to be valid if we have at least $2m+1$ generic time points by Sontag's result \cite{sontag2002differential}, but can be valid with fewer generic time points. Indeed, as the Linear ERK model admits analytic solutions, we can prove that 
it is globally structurally identifiable for any choice of three distinct time points. Determining structural identifiability for specific time points in the absence of analytic solutions is an open problem.
We numerically show that the Full ERK model and Rational ERK model are not practically identifiable; however, the source of this practical non-identifiability is not completely clear (see Section 4).

Finally, for a model that is structurally and practically identifiable, one would like to infer parameters, i.e., what parameter values are consistent with the observations? We perform Bayesian inference, as done in \cite{Yeung2020}, and extend this to the Rational ERK model. 
The result of the parameter inference on the Linear ERK model is a sample point cloud of posterior densities of inferred ERK parameter kinetics that are consistent with the data. We obtain five different sample densities corresponding to the five MEK variants.

In Section~\ref{sec:topology}, we compare the geometry of the admissible regions of parameter space of the five MEK variants. We implement a theoretical framework originally proposed by Taylor \etal \cite{Taylor2019} to quantify the shape of the resulting posterior distributions using topological data analysis and facilitate a comparison between mutants. Specifically, we approximate the persistent homology of super-level sets of posterior densities by simplicial complexes. We perform these measurements on the distributions obtained from Bayesian parameter inference for 
the 5 MEK variants and compare them via a topological bottleneck distance.

\begin{result}
The topological data analysis quantifies that the Linear ERK model parameter posteriors are most different between the WT and SSDD mutant data. 
The kinetics of the SSDD mutant, which mimics phosphorylated MEK, has the largest topological distance from all other MEK/ERK mutants. 
\end{result}

This biological result raises the question of whether the SSDD variant is a suitable approximation for wild-type MEK activated by Raf, and suggests further experimental studies are needed. While the previous analysis by Yeung \etal \cite{Yeung2020} compared the variants by the inferred kinetics of each parameter, here we complement that analysis by comparing the three parameters together as a point cloud.

In summary, our aim is to showcase how systematic algebraic, geometric and topological approaches can be applied to a biologically relevant model with state-of-the-art experimental time-course data. Each of these approaches incorporates the structure of the mathematical model, experimental observations, and \textit{experimental setup and observations} (e.g., experimental initial condition, observable species, number of experimental replicates, number of time points collected, etc), as well as known biological information (e.g., published parameter values). The framework is not limited to this case study and may enhance the analysis of models in systems and synthetic biology.

\section{From ERK biochemical reactions to a polynomial dynamical system}\label{sec:model}

Protein phosphorylation alters protein function in signalling pathways and plays a crucial role in cellular decisions and homeostasis. Phosphorylation is the addition of a phosphate group by an enzyme known as a kinase, and dephosphorylation is the removal of a phosphate group by an enzyme known as a phosphatase. Multisite phosphorylation is the process of having multiple possible locations on a protein phosphorylated, which increases the number
of potential ways protein function can be altered.
The algebra, geometry, combinatorics and dynamics of multisite phosphorylation has been a source of interesting mathematical problems \cite{dickenstein2016biochemical,manrai2008geometry,conradi2019multistationarity}.
A protein with $q$ phosphorylation sites has been shown to have $2^q$ phospho-states; the sites on the protein can be phosphorylated in $q!$ possible ways \cite{thomson2009unlimited}. One of the simplest multisite phosphorylation systems is when a protein has two phosphorylation sites. We focus on the sequential dual phosphorylation of the extracellular signal regulated kinase (ERK) by its kinase activated (dually phosphorylated) MEK. The model developed by Yeung et al. \cite{Yeung2020} encodes a mixed phosphorylation mechanism (i.e., distributive and processive) by changes in parameter values rather than separate models (see e.g., \cite{conradi2015global,gunawardena2007distributivity} and references therein). This enabled them to quantify the extent to which a MEK variant is processive or distributive. We remark that the model presented by Yeung et al. does not include dephosphorylation mechanisms, since the experimental setup omitted the addition of phosphatases. 

Next, we introduce the model and the experimental data published by Yeung et al. \cite{Yeung2020}.

\subsection{The Model}

The protein \textit{substrate} ERK, is activated through dual phosphorylation by its activated \textit{enzyme} kinase MEK. As shown in the chemical reaction network (see Figure \ref{fig:crn}), unphosphorylated ERK ($S_0$) binds reversibly with its kinase MEK ($E$) to form an intermediate complex $C_1$. The complex becomes $C_2$ when a phosphate group is added.
Complex $C_2$ can then disassociate to form MEK ($E$) and ERK phosphorylated on the tyrosine site ($S_1$), or a second phosphate group is added to $C_2$, resulting in product reactants $E+S_2$. The six species and six rate constants are given in the following chemical reaction network (Figure~\ref{fig:crn}).

\begin{figure}[h!]
\centering
\begin{tikzpicture}[baseline= (a).base]
\node[scale=1] (a) at (0,0){
\begin{tikzcd}
E+S_0 \arrow[r, "k_{f_1}", shift left=0.6ex] & C_1 \arrow[l, "k_{r_1}", shift left=0.6ex] \arrow[r, "k_{c_1}"] & C_2 \arrow[r, "k_{c_2}"] \arrow[d, "k_{r_2}"', shift right=0.6ex] & E+S_2 \\
                           &                                                  & E+S_1 \arrow[u, "k_{f_2}"', shift right=0.6ex]                      &      
\end{tikzcd}
};
\end{tikzpicture}
\caption{The reaction network associated with dual phosphorylation of ERK by its activated enzyme kinase MEK.}
\label{fig:crn}
\end{figure}
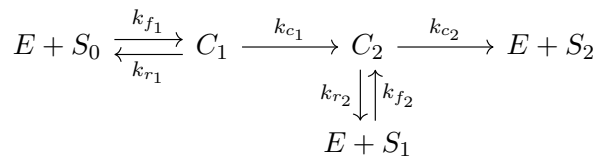
\glsxtrnewsymbol[description={Concentration of species $i$}]{Si}{\ensuremath{S_i}}
\glsxtrnewsymbol[description={Concentration of compound $i$}]{Ci}{\ensuremath{C_i}}
\glsxtrnewsymbol[description={Concentration of free enzyme}]{E}{\ensuremath{E}}
\glsxtrnewsymbol[description={Total substrate concentration}]{Stot}{\ensuremath{S_{tot}}}
\glsxtrnewsymbol[description={Total concentration of free enzyme}]{Etot}{\ensuremath{E_{tot}}}
\glsxtrnewsymbol[description={Reaction rate of reaction step $j$}]{kj}{\ensuremath{k_{j}}}
\glsxtrnewsymbol[description={Tuple of all model parameters}]{theta}{\ensuremath{\theta}}
\glsxtrnewsymbol[description={Tuple of all species concentrations in a model}]{x}{\ensuremath{x}}

We can translate this reaction network into a dynamical system $\dot{x} = f(x,\theta)$. Here, $f$ is a vector-valued function of the vectors of species concentrations $x= \{S_0, C_1, C_2, S_1, S_2,E\}$ and rate constants, referred to as parameters $\theta = \{k_{f_1},k_{r_1},k_{c_1},k_{f_2}, k_{r_2},k_{c_2} \}$. The kinetics assumption for $f$ is a modelling choice; here we assume that the law of mass action holds \cite[\S 2.1.1]{klipp2016systems}, as for the original model \cite{Yeung2020}. The resulting dynamical system of ODEs is given in Equations (\ref{eq:a}).

\begin{subequations}\label{eq:a}
\begin{flalign}
\frac{\mathrm dS_0}{\mathrm dt}&= -k_{f_1}E\cdot S_0+k_{r_1}C_1\label{eq:a_start},\\
\frac{\mathrm dC_1}{\mathrm dt}&= k_{f_1}E\cdot S_0 -(k_{r_1}+k_{c_1})C_1, \label{eq:ab}\\
\frac{\mathrm dC_2}{\mathrm dt}&= k_{c_1}C_1-(k_{r_2}+k_{c_2})C_2 + k_{f_2}E\cdot S_1,\label{eq:ac}\\
\frac{\mathrm dS_1}{\mathrm dt}&= -k_{f_2}E\cdot S_1+k_{r_2}C_2,\label{eq:ad}\\
\frac{\mathrm dS_2}{\mathrm dt}&=k_{c_2}C_2,\label{eq:ae}\\
\frac{\mathrm dE}{\mathrm dt}&= -k_{f_1}E\cdot S_0+k_{r_1}C_1-k_{f_2}E\cdot S_1 +(k_{r_2}+k_{c_2})C_2.\label{eq:a_end}
\end{flalign}
\end{subequations}
We assume that initially all species are zero, except for $S_0(t=0)=S_{tot}$ and $E(t=0)=E_{tot}$. 
 Equations~(\ref{eq:c1})-(\ref{eq:c2}) define two conserved quantities that constitute a basis for the linear space of conservation relations of the model:
\begin{flalign}
S_{tot}&=S_0+S_1+S_2+C_1+C_2,\label{eq:c1}\\
E_{tot}&=E+C_1+C_2,\label{eq:c2}
\end{flalign}
where the total amounts of substrate ERK ($S_{tot}$) and enzyme MEK ($E_{tot}$) are constant and known from the initial conditions. 

We aim to study the relationship between the species $x$, parameters $\theta$, conserved quantities, and available biological information (previous knowledge and experimental observations).
The emphasis in this paper is not to analyse the steady-state variety as in \cite{gross2016algebraic}, rather here we focus on the transient dynamics of the model and algebraic approaches to analyse ERK kinetics in light of the available biological information.

\subsection{The Data}\label{sec:data}
The data is published. Details on measurement techniques and experimental methods can be found in \cite{Yeung2020}. We present the experimental setup for the data we analyse.

\subsubsection{Experimental setup and data}
In all experiments, $0.65\mu M$ free (activated) enzyme MEK $(E)$ was added to $5\mu M$ of unphosphorylated ERK substrate $(S_0)$ along with ATP; therefore, $S_{tot}=5\mu M$ and $E_{tot}=0.65\mu M$.  ERK was measured in three states: unphosphorylated $(S_0+C_1)$, mono-phosphorylated $(S_1+C_2)$, and dually-phosphorylated ERK $(S_2)$, at 7 time points, with $r$ different experimental replicates.
The sample space for each MEK variant is
$X=\RR^{3\times 7\times r}$,
where: for human wild-type MEK, $r=11$; for 
MEK variants with phosphomimetic (SSDD), $r=6$; and for activating mutations, $r=5$. The three activating mutants of MEK are known to be involved in human cancer (E203K) or developmental abnormalities (F53S and Y130C). The ERK observations were collected at seven time points $t=\{0.5, 2, 3.25, 3.75, 5, 10, 20\}$ minutes for all MEK variants except SSDD, which were collected at $t= \{1, 2, 3.25, 5, 10, 20, 40\}$ minutes. 

\glsxtrnewsymbol[description={Tuple of all measurements of species concentrations in a set of experiments}]{X}{\ensuremath{X}}
\glsxtrnewsymbol[description={Time}]{t}{\ensuremath{t}}
\glsxtrnewsymbol[description={Michalis-Menten constant associated to compound $C_i$}]{MM}{\ensuremath{k_{M_i}}}

\subsubsection{Known biological information}
The relationship between some kinetic rate constants is known. When a substrate binds reversibly to an enzyme to form an enzyme-substrate complex, which then reacts irreversibly to form a product and release the enzyme, one can define the Michaelis-Menten constant $k_{M}$. 
In the reaction network given by Equations (\ref{eq:a}), there are two Michaelis-Menten constants
$k_{M_i}=(k_{c_i}+k_{r_i})/k_{f_i}$ for $i=1,2$. Measurements show that in our experimental setup $k_{M_i}\approx 25\mu M$ for $i=1,2$ \cite{Taylor15514}. While the reaction rates $k_{c_i}$ and $k_{r_i}$ for $i=1,2$ cannot be measured directly, they have been shown to be the same order of magnitude \cite{BarEven11}. We will use these insights to assume, henceforth, that $S_0$, $S_1$, and $S_2$ were measured (without added compound variables). We justify this mathematically in Subsection \ref{sec:output_var}.

\section{Algebraic Model Reduction}\label{sec:algred}

The first step to studying most models typically involves model reduction, which reduces the number of dependent variables and constant parameters. For many chemical reactions, there are time scales on which the rate of change of some variables is negligible and their 
dynamics is dominated by those of the remaining variables. 
This observation motivates the 
\emph{Quasi-Steady-State-Approximation} (QSSA).

Classical QSSA dates back to work by Henri, Michaelis-Menten, and Briggs and Haldane who analysed chemical reactions, using heuristic arguments based on fast and slow reactions, and assuming initial enzyme concentrations are much smaller than substrate concentrations. In the 1960s, a mathematical framework for QSS reduction using singular perturbation theory was developed by Heineken et al. \cite{Heineken1967}, 
which enabled rigorous convergence proofs. To determine parameter regions in which QSSA holds, the singular perturbation theory approach inspired the more prevalent ``slow-fast" timescale arguments in the seminal work of Segel and Slemrod \cite{segslem}. Another approach, proposed in 1983 by Schauer and Heinrich \cite{Schauer1983},
justified the Michaelis-Menten procedure mathematically by requiring that relevant trajectories of the full model remain close to the quasi-steady-state trajectories
-- also referred to in the algebraic literature as the ``QSS variety". 

In recent years, algebraic approaches to reduce polynomial ODE models have been extended by Walcher
and other applied algebraic geometers. In 2013, Pantea et al. \cite{pantea} used Galois theory to characterise chemical reaction networks for which no explicit QSSA reduction is possible. Furthermore, they provided computational tools for determining the feasibility of an explicit reduction.
Subsequently, Sweeney \cite{sweeney2017conditions} proved that the nonsolvability of polynomials poses no issue to the CRNs most commonly encountered in practice and derived a more efficient algorithm for determining explicit reducibility by translating algebraic structures into graphs. Goeke and Walcher (2014) \cite{Goeke2014}
provide an explicit formula for obtaining a reduced QSSA model using a subset of an algebraic variety defined by the slow manifold.  Subsequently, Goeke et al. (2017) \cite{reduction} characterised parameter values at which QSSA reduction is accurate using algebraic varieties and bounds on the polynomials governing the ODE system on a bounded parameter- and variable-domain. 
Most recently, 
Feliu et al. \cite{feliu2019quasisteady} derived
necessary and sufficient conditions for purely algebraic reductions of a CRN model to agree with 
model reductions derived via classical singular-perturbation theory \cite{Keener2011, Segel1988}.

In this Section, we briefly review QSSA using classical singular-perturbation theory as well as the algebraic approaches developed by Goeke et al. \cite{reduction, Goeke2014}. We then apply both methods to the full ERK model (Equations (\ref{eq:a})).
We show both approaches can generate the same QSSA-reduction of our ERK model, which we will call the Rational ERK model. Additionally, the algebraic method can yield a linear QSSA-reduction of our ERK model in a single step (which we call the Linear ERK model). By contrast, the singular-perturbation-theory approach requires additional assumptions on parameter values to arrive at the Linear ERK model (see Subsection \ref{sec:algebraic_erk}). 
We show that the Linear ERK model approximates the Full ERK model (Equations (\ref{eq:a})) with similar accuracy as the Rational ERK model in the context of the experimental setup, data and known biological information (see Subsection \ref{sec:algebraic_erk}; Appendix \ref{sec:accuracy} for details).

With the algebraic method, we provide a rigorous mathematical justification of the Linear ERK model presented by Yeung et al \cite{Yeung2020}. By comparing the singular perturbation method with the algebraic method and the two resulting model reductions, we 
illustrate how the algebraic methods form a well-structured
approach for arriving at a QSS reduction and for assessing the accuracy of such reductions systematically.

\subsubsection*{Notation for Model Reduction}
Throughout, we will assume we have an ODE system in variables $x\in\RR^n$ and parameters $\theta\in\RR^m$. 
If the system dynamics are governed by $f$, a vector of polynomials in $\RR[x,\theta]^n$, then our ODE system is given by
\begin{equation}
\frac{\mathrm{d} x}{\mathrm{d}t} = f(x,\theta). \label{eq:ode}
\end{equation}
For $1\leq q<n$, we may define
\begin{flalign}\notag
&\quad x^{[1]}=(x_1,...,x_q),&&\quad f^{[1]}=(f_1,...,f_q),&\\\notag
&\quad x^{[2]}=(x_{q+1},...,x_n),&&\quad f^{[2]}=(f_{q+1},...,f_n).&
\end{flalign}
We wish to retain the variables $x^{[1]}$ in the reduced model and seek to eliminate variables $x^{[2]}$ as part of our model reduction.

For the full ERK model (Equations (\ref{eq:a})), we choose $x^{[1]}:=(S_0, S_1, S_2)$ and $x^{[2]}:=(C_1, C_2)$. Analogously, $f^{[1]}$ are the polynomials governing the rates of change of $S_0$, $S_1$ and $S_2$ (Equations (\ref{eq:a_start}), (\ref{eq:ad}) and (\ref{eq:ae}) and $f^{[2]}$ are the polynomials governing the rates of change of $C_1$ and $C_2$ (Equations (\ref{eq:ab}) and (\ref{eq:ac})).
\glsxtrnewsymbol[description={Rates of change of concentrations given $x$ and $\theta$}]{f}{\ensuremath{f(x,\theta)}}

\begin{remark*}
In the current section, we treat the (non-zero) initial conditions of the ODE systems as parameters (and include them in the parameter count $m$), as they are central to determining the goodness of a model reduction. In Section \ref{sec:ident} (Identifiability) and Section \ref{sec:topology} (Inference \& Comparison), we will not include the initial conditions in the set of  parameters, as they are given by the experimental setup and, as such, do not need to be identified or inferred.
\end{remark*}

%------------------
%-----------------

\subsection{The Classical Singular-Perturbation Theory Approach}\label{sec:spt}

The fundamental assumption for the application of a QSSA as presented in \cite{segslem} is the existence of a separation of time scales. This means, there exist timescales
$0<t_S<t_L$ (called the short and long timescale, respectively) such that the variables in $x^{[1]}$ exhibit significant variation only when $t>t_L$ while the variables in $x^{[2]}$ exhibit rapid variation when $0<t<t_S$. By contrast, the variation of $x^{[2]}$ is dominated by the variables $x^{[1]}$ when $t>t_L$. 
From a biological perspective, natural choices for these timescales of the Full ERK model are $t_S:= (S_{tot} k_{f_1})^{-1}$ and $t_L:= (E_{tot} k_{f_1})^{-1}$.

The next step is to analyse the model on the long time scale. We introduce dimensionless variables $\tilde{x}^{[1]}:=x^{[1]}/B_1$, $\tilde{x}^{[2]}:=x^{[2]}/B_2$) for constants $B_1$, $B_2$ and rescale time so that $T:=t/t_L$.
Then, the system can be written as
$$\frac{\mathrm d\tilde{x}^{[1]}}{\mathrm dT}=\tilde{f}^{[1]}(\tilde{x},\theta),\qquad \varepsilon\frac{\mathrm d\tilde{x}^{[2]}}{\mathrm dT}=\tilde{f}^{[2]}(\tilde{x},\theta),$$
where $\tilde{f}^{[1]}$ and $\tilde{f}^{[2]}$ are appropriate polynomials, and 
$0 < \varepsilon = t_S/t_L \ll 1$ is the ratio of the timescale on which the variables in $x^{[2]}$ exhibit
rapid variation to the timescale on which the variables in $x^{[1]}$ exhibit
significant variation. 
%%%
The reduction is performed by assuming $\varepsilon\to0$, which gives rise to $0=\tilde{f}^{[2]}(\tilde{x},\theta)$. 
This result can be exploited to obtain expressions
for $\tilde{x}^{[2]}$ in terms of $\tilde{x}^{[1]}$.
The reduced system 
has $q<n$ model variables 
and, typically, the dimension of the parameter space is also decreased. 

\glsxtrnewsymbol[description={Non-dimensional concentrations}]{xtilde}{\ensuremath{\tilde{x}}}
\glsxtrnewsymbol[description={Rates of change of non-dim. concentrations given $\tilde{x}$ and $\theta$}]{ftilde}{\ensuremath{\tilde{f}(\tilde{x},\theta)}}

\glsxtrnewsymbol[description={Time variable in fast timescale}]{T}{\ensuremath{T}}

The classical QSSA approach applied to the Full ERK model is presented in Appendix \ref{sec:qssa_mr}, and yields the Rational ERK model.

\subsection{The Algebraic QSSA Approach}\label{sec:alg_qssa}
The algebraic approach to QSSA, as presented by Goeke, Walcher and Zerz in \cite{reduction}, differs from the classical approach in several ways. 
Most notably, an \emph{a priori} separation of time scales is not needed. On the other hand, we require a choice of fast and slow variables (i.e., a choice of which variables we eliminate from, 
and which we retain in 
the reduced model).

\begin{remark*}
To the best of our knowledge, 
all existing algebraic approaches to QSSA, including \cite{Goeke2011,Goeke2014, boulier2011model,reduction}, require a choice, explicit or implicit, of slow and fast variables. In \cite{Goeke2011} the relevant choice is made by expanding $f$ (in \cite{Goeke2011}: $h$) at different orders of $\varepsilon$.
\end{remark*}

First, we characterise points in parameter space, i.e., parameter values, where the fast variables are exactly determined by the slow variables, which yields a reduced model. This set of parameter values is defined as the vanishing set of the polynomials governing the ODEs of the fast variables.
This defines an algebraic variety in the parameter space. Typically, the ODE system will be degenerate at these values. Secondly, we 
characterise neighbourhoods of these values in parameter space, as well as timescales for which the reduction is a good approximation to the original model.

To describe the characterisation from \cite{reduction}, we use $x^{[1]}$, $x^{[2]}$, $f^{[1]}$, and $f^{[2]}$ as before. In addition, we denote the partial derivative with respect to $x^{[i]}$ by $D_i$. For a fixed $\theta^*\in\RR^m$, we let $Y_{\theta^*}$ denote the algebraic variety defined by $f^{[2]}(\,\cdot\,,\theta^*)$.

\begin{definition}\label{def:qssvar}
Let $y\in Y_{\theta^*}$ be such that the $(n-q)\times (n-q)$ matrix $D_2 f^{[2]}$ has full rank at $(y,\theta^*)$. Then we denote by $V_{\theta^*}\subseteq Y_{\theta^*}$ a relatively Zariski-open neighbourhood of $y$ in which this rank is maximal. We call $V_{\theta^*}$ a \emph{quasi-steady-state (QSS) variety} in the sense of \cite{reduction} and may assume without loss of generality that it is irreducible.

If, furthermore, $V_{\theta^*}$ is an invariant set of the ODE system ${\mathrm dx^{[1]}/\mathrm dt}=f(x,\theta^*)$, then we call $\theta^*$ a \emph{QSS parameter value}. Recall that in dynamical systems theory, $V_{\theta^*}$ is an invariant set of $\RR^q$ if whenever the initial condition of an ODE at $t=0$
is in $V_{\theta^*}$, then the corresponding trajectories of the ODE remain in $V_{\theta^*}$ for all $t>0$.
\end{definition}

\begin{remark*}
Note that the steady-state variety (see \cite{gross2016algebraic}) and the QSS variety at a parameter value $\theta^*$ are not as closely related as one may first think. Indeed with our notation, the steady state variety is the zero set in $\RR^n\times\RR^m$ of the ideal $\langle f^{[1]}(x,\theta),f^{[2]}(x,\theta)\rangle$ of $\RR[x,\theta]$, while the QSS variety at $\theta^*$ is contained in the zero set in $\RR^n\times\{\theta^*\}$ of the ideal $\langle f^{[2]}(x,\theta^*)\rangle$ of $\RR[x]$. That is, we have both $\mathcal{V}_{\RR^n\times\RR^m}(f^{[1]}(x,\theta),f^{[2]}(x,\theta))\subset \mathcal{V}_{\RR^n\times\RR^m}(f^{[2]}(x,\theta)$ and $V_{\theta^*}\subseteq Y_{\theta^*}=\mathcal{V}_{\RR^n\times\{\theta^*\}}(f^{[2]}(x,\theta^*))\subset \mathcal{V}_{\RR^n\times\RR^m}(f^{[2]}(x,\theta))$, but the steady-state variety and $V_{\theta^*}$ 
are not contained in one another in general.
\end{remark*}

\glsxtrnewsymbol[description={Parameter / QSSA-variety of $\theta$}]{variety}{\ensuremath{Y_\theta,\, V_\theta}}

To apply the theory of Goeke, Walcher and Zerz in \cite{reduction}, we assume that the initial condition of our ODE system (Eq. (\ref{eq:ode})) lies in $V_{\theta^*}$.
As $D_2f^{[2]}$ has full rank on $V_{\theta^*}$, we have that $x^{[2]}=\Psi\left(x^{[1]}\right)$ for some continuous $\Psi$ by the Implicit Function Theorem. Hence, writing $x=(x^{[1]}, x^{[2]})$, we obtain a \emph{reduced model}:
\begin{equation}\frac{\mathrm dx^{[1]}}{\mathrm dt}=f^{[1]}\left(\left(x^{[1]}, \Psi\left(x^{[1]}\right)\right), {\theta^*}\right)\label{eq:impfunthm}\end{equation}
on some open neighbourhood in $\RR^j$ that naturally includes
$V_{\theta^*}$.
This corresponds to determining the fast variables in terms of the slow variables. We do this by setting their time rates of change equal to zero on the short timescale in classical QSSA, with the addition that on $V_{\theta^*}$ the above yields an exact solution rather than an approximation. 
As a caveat, we note that, in both settings, it may not be possible to find an algebraic expression for $\Psi$; this was pointed out and completely characterised by Pantea et al. in \cite{pantea} in terms of Galois theory. Because of the possible non-solvability issue with Equation (\ref{eq:impfunthm}), we require a more general methodology (Proposition \ref{prop:red}) to study the accuracy of a model reduction (Proposition \ref{prop:accuracy}).

Goeke, Walcher and Zerz
showed that locally, in the variable $x^{[1]}$, the reduced system given by Equation (\ref{eq:impfunthm}) has the same solution as the following ODE system
\begin{equation}
\frac{\mathrm dx^{[1]}}{\mathrm dt}=f^{[1]}\left(x, {\theta^*}\right),\qquad
\frac{\mathrm dx^{[2]}}{\mathrm dt}=-D_2 f^{[2]}(x, {\theta^*})^{-1}D_1f^{[2]}(x, {\theta^*})f^{[1]}(x, {\theta^*}):\label{eq:expred}
\end{equation}

\begin{proposition}[Lemma 1 \& Proposition 1 in \cite{reduction}]\label{prop:red}
Let $V_{\theta^*}$ be a QSS-variety. Then $V_{\theta^*}$ is an invariant set of Equation (\ref{eq:impfunthm}). Moreover, any solution of Equation (\ref{eq:expred}) with initial condition in $V_{\theta^*}$ locally solves Equation (\ref{eq:impfunthm}). Conversely, any solution of Equation (\ref{eq:impfunthm}) with initial condition in $V_{\theta^*}$ locally solves Equation (\ref{eq:expred}). 
In addition, $V_{\theta^*}$ is an invariant set of Equation (\ref{eq:ode}) if and only if the solutions of Equations (\ref{eq:ode}) and (\ref{eq:expred}) are equal for all initial conditions in $V_{\theta^*}$.
\end{proposition}

This proposition equips us with a method to obtain a solution for $x^{[1]}$ in an algebraic QSSA without explicitly determining $\Psi$. 
In Sections \ref{sec:ident} and \ref{sec:topology},
we will use Equation (\ref{eq:impfunthm}) as our model reduction. 

First, however, we assess the accuracy of Equation (\ref{eq:expred}) as an approximation to the full system, 
for parameter-values $\theta$ in some neighbourhood of $\theta^*$. For convenience, we abbreviate system (\ref{eq:expred}) as ${\mathrm dx/\mathrm dt}=f_\mathrm{red}(x,\theta^*)$.

\begin{proposition}[Outline of Proposition 2 in \cite{reduction}]\label{prop:accuracy}
Let $K^*\subseteq\RR_+^n\times\RR_+^m$ be a compact domain in the product of the variable and parameter spaces which satisfies a number of conditions (we refer the interested reader to Appendix \ref{sec:accuracy} for details). Let ${\theta^*}$ be given such that $V_{\theta^*}\times\{{\theta^*}\}$ has non-empty intersection with $\mathrm{int}\,K^*$, let $(y,{\theta^*})$ be a point in this intersection, and let $V'_{\theta^*}$ be an open neighbourhood of $y$ such that $(V_{\theta^*}\cap V'_{\theta^*})\times\{{\theta^*}\}\subseteq K^*$. Additionally, let $t^*>0$ be such that the solution of Equation (\ref{eq:ode}), with initial condition $y$, remains in $V'_{\theta^*}$ for $t\in[0,t^*]$.

Then there exists a compact neighbourhood $A_{\theta^*}\subseteq V_{\theta^*}$ of $y$ such that:
\begin{enumerate}
\item[(i)] For every $z\in A_{\theta^*}$, the solution of Equation (\ref{eq:ode}) with initial condition $z$ exists and remains in $V'_{\theta^*}$ for $t\in[0, t^*]$.
\item[(ii)] For every $\varepsilon'>0$, there exists a $\delta_1>0$ such that for every $z\in V'_{\theta^*}\cap A_{\theta^*}$ the solution of Equation (\ref{eq:expred}), with initial condition $z$, exists and remains in $V'_{\theta^*}$ for $t\in[0,t^*]$ whenever $\|f-f_\mathrm{red}\|<\delta_1$ on $V_{\theta^*}$.
\item[(iii)] For every $\varepsilon'>0$, there exists a $\delta\in(0,\delta_1]$ such that, for any $z\in V_{\theta^*}\cap A_{\theta^*}$, the difference between the solutions of Equations (\ref{eq:expred}) and 
(\ref{eq:ode}), with initial condition $z$, is at most $\varepsilon'$ for $t\in[0,t^*]$ whenever $\|f-f_\mathrm{red}\|<\delta$ on $V'_{\theta^*}$. Here, $\|\,\cdot\,\|$ denotes the infinity-norm over the interval $[0, t^*]$ for a fixed parameter value.
\end{enumerate}
\end{proposition}

In summary, given some technical assumptions on the variables 
and the domain $K^*$, we can bound the difference between the solutions of Equations (\ref{eq:ode}) and (\ref{eq:expred}) in terms of $\|f-f_\mathrm{red}\|$ up to some time $t^*>0$. The full statement of this proposition also includes  
lower bounds on this difference. Note that we do not assume that $\theta^*$ is a QSS-parameter value, but the assumptions on $K^*$ (as detailed in Appendix \ref{sec:accuracy}) require it to be close to some QSS-parameter value.

\subsection{Reducing the ERK Model Algebraically}\label{sec:algebraic_erk}
We now apply the theory from Subsection \ref{sec:alg_qssa} to the Full ERK model 
(Equations (\ref{eq:a})) in two different ways, to derive
two reduced models (the Linear ERK and Rational ERK models). The full details of the derivations can be found in Appendix \ref{sec:algebraic_mr}. We also give a brief biological explanation of why both systems explain the phenomena underlying the given experimental data equally well.

\subsubsection{Reduction via conservation laws} We can exploit the conservation laws (\ref{eq:c1}) and (\ref{eq:c2}) to eliminate a variable
before using
the analytic or algebraic QSSA approach. First, we choose to eliminate $E$ and note that there are two choices:
\begin{equation}
E=E_{tot}-C_1-C_2 \label{eq:e1}
\end{equation}
or
\begin{equation}
E=E_{tot}-S_{tot}+S_0+S_1+S_2.\label{eq:e2}
\end{equation}

Subsequently, we choose to eliminate the variables $C_1$ and $C_2$ via (algebraic) QSSA. For the Rational ERK model, 
using (\ref{eq:e1}) to eliminate $E$, we obtain
$$f_\mathrm{rat}^{[2]}=\begin{pmatrix}k_{f_1}(E_{tot}-C_1-C_2)\cdot S_0 -(k_{r_1}+k_{c_1})C_1\\ k_{c_1}C_1 + k_{f_2}(E_{tot}-C_1-C_2)\cdot S_1 -(k_{r_2}+k_{c_2})C_2\end{pmatrix},$$
while for the Linear ERK model, employing substitution (\ref{eq:e2}), we have
$$f^{[2]}_\mathrm{lin}=\begin{pmatrix}k_{f_1}(E_{tot}-S_{tot}+S_0+S_1+S_2)\cdot S_0 -(k_{r_1}+k_{c_1})C_1\\ k_{c_1}C_1 + k_{f_2}(E_{tot}-S_{tot}+S_0+S_1+S_2)\cdot S_1 -(k_{r_2}+k_{c_2})C_2\end{pmatrix}.$$

\subsubsection{Reduction via an algebraic QSSA} To reduce the model further, we apply an algebraic QSSA, as described in Subsection \ref{sec:alg_qssa}. We start by identifying QSS-parameter-values. For 
$f^{[2]}_\mathrm{rat}$, we have
$$D_2f^{[2]}_\mathrm{rat}=\begin{bmatrix}-k_{f_1}S_0-(k_{r_1}+k_{c_1}) & -k_{f_1}S_0 \\ -k_{f_2}S_1+k_{c_1} & -k_{f_2}S_1-(k_{r_2}+k_{c_2})\end{bmatrix},$$
while for $f^{[2]}_\mathrm{lin}$ we have
$$D_2f^{[2]}_\mathrm{lin}=\begin{bmatrix}-(k_{r_1}+k_{c_1}) & 0 \\ k_{c_1} & -(k_{r_2}+k_{c_2})\end{bmatrix}.$$
In both cases, assuming that $(k_{r_i}+k_{c_i})>0$ for $i=1,2$ (otherwise, the reaction network would be degenerate, meaning some or all variables would remain constant), 
and given that $S_0$ and $S_1$ are non-constant, 
we deduce that these matrices are 
invertible.
Hence, both substitutions (\ref{eq:e1}) and (\ref{eq:e2}) are good candidates for an algebraic QSSA reduction.

We note that the assumption $E_{tot}=0$ is required to ensure that the initial condition lies in $V_{\theta^*}$. This is not physically realistic, as the absence of free enzyme makes the reaction rates negligible, however, in parameter space this assumption is close to the experimental setup ($E_{tot}\approx 0.65\mu M$). In fact, unlike the rate parameters, we know the value of $E_{tot}$ and can, therefore, bound the error associated with such an idealisation (cf. Appendix \ref{sec:accuracy}). The assumption that $E_{tot}=0$ is similar to the classical singular-perturbation theory approach, where a typical choice of short timescale is $(t_{S}=E_{tot}k_{f_1})$ and one then subsequently assumes $\varepsilon= E_{tot}/S_{tot}\to0$.

As $E_{tot}=0$ will yield a stationary model and ensure that $V_{\theta^*}$ contains the initial condition, we find that any parameter value $\theta^*$ satisfying $(k_{r_i}+k_{c_i})>0$ for $i=1,2$ and $E_{tot}=0$ is a QSS-parameter-value for both the Rational and Linear ERK model.

For both models, we have
$$Y_{\theta^*}=\left\{x=(S_0,S_1,S_2,C_1,C_2)\in\RR^5\,\vert\,f^{[2]}(x,\theta^*)=0\right\}.$$

For the Linear ERK model, we can show that $Y_{\theta^*}^\mathrm{lin}$ is irreducible (at generic parameter values) and thus its QSS-variety is $V_{\theta^*}^\mathrm{lin}=Y_{\theta^*}^\mathrm{lin}$. For the Rational ERK model, we have that
 $Y_{\theta^*}^\mathrm{rat}$ decomposes as
$$Y_{\theta^*}^\mathrm{rat}=(Y_{\theta^*}^\mathrm{rat}\cap\mathcal{V}(\langle C_1+C_2\rangle))\cup
(Y_{\theta^*}^\mathrm{rat}\cap\mathcal{V}(\langle \lambda(k_{r_2}+k_{c_2})+S_0+\lambda k_{f_2}S_1\rangle))$$
where $\lambda:=-k_{r_1}/(k_{f_1}(k_{c_1}-k_{c_2}-k_{r_1}))$. At generic parameter values, only the first irreducible component will contain the initial condition. Hence, the natural choice for the QSS-variety is $$V^\mathrm{rat}_{\theta^*}=\left\{x=(S_0,S_1,S_2,C_1,C_2)\in\RR^5\,\vert\,C_1=0,\,C_2=0\right\}.$$

The substitution (\ref{eq:e1}) yields the \emph{Rational ERK model} given by
\begin{subequations}\label{eq:q}
\begin{flalign}
\frac{\mathrm dS_0}{\mathrm dt}&= \frac{-\kappa_1S_0}{ \gamma_1 S_0+\gamma_2S_1+1},\label{eq:q1}\\
\frac{\mathrm dS_1}{\mathrm dt}&=\frac{-\kappa_2S_1+(1-\pi)\kappa_1S_0}{ \gamma_1S_0+\gamma_2S_1+1},\\
\frac{\mathrm dS_2}{\mathrm dt}&=\frac{\pi\kappa_1S_0 + \kappa_2S_1}{ \gamma_1S_0+\gamma_2S_1+1},\label{eq:q3}
\end{flalign}
\end{subequations}
while the substitution (\ref{eq:e2}) gives the \emph{Linear ERK model}:
\begin{subequations}\label{eq:l}
\begin{flalign}
\frac{\mathrm dS_0}{\mathrm dt}&= {-\kappa_1S_0},\label{eq:l1}\\
\frac{\mathrm dS_1}{\mathrm dt}&={-\kappa_2S_1+(1-\pi)\kappa_1S_0},\\
\frac{\mathrm dS_2}{\mathrm dt}&= \pi\kappa_1S_0 + \kappa_2S_1.\label{eq:l3}
\end{flalign}
\end{subequations}
Here, for $i=1,2$, we use the newly introduced quantities
\begin{equation}\kappa_i=E_{tot}\frac{k_{f_i}k_{c_i}}{ k_{c_i}+k_{r_i}},\qquad \pi=\frac{k_{c_2}}{ k_{c_2}+k_{r_2}},\qquad \gamma_i=k_{f_i}\frac{k_{c_1}+k_{c_2}}{\left(k_{c_1}+k_{r_1}\right)\left(k_{c_2}+k_{r_2}\right)}.\label{eq:pi1pi2kappa}\end{equation}
\glsxtrnewsymbol[description={Parameters of reduced models}]{kappa}{\ensuremath{\kappa_i, \pi, \gamma_i}}
Both models are reductions obtained via the ODE system (\ref{eq:impfunthm}).
The processivity parameter, which is the probability that
both phosphorylations are carried out by the same
enzyme, is represented by $\pi$ in the reduced models. The $\kappa_i$ represents the kinetic efficiencies of the first and second phosphorylation steps, respectively \cite{Yeung2020}.

It should be noted that the Rational ERK model is the system we would obtain via the classical singular perturbation approach. In Appendix \ref{sec:qssa_mr}, we explain how to arrive at these equations following the analytic approach, while in Appendix \ref{sec:algebraic_mr} we detail how to derive both model reductions algebraically.

\subsubsection{Assessing Accuracy}\label{subsec:qssa-accuracy}

We can use the algebraic framework of Goeke, Walcher and Zerz and, in particular, Proposition \ref{prop:accuracy} to bound the error of the Linear ERK model reduction to the full model.  Given the measurements of the Michaelis-Menten constants $k_{M_i}$, we can derive simple expressions which bound the approximation error (see Appendix \ref{sec:accuracy} for both the Rational \& Linear ERK model). 
Unfortunately, the bound on the approximation error
depends on parameters with unknown values. However, we can compare the bounds derived for the Linear ERK model to those for the Rational ERK model and show that in the regime where $k_{M_i}\approx 25\mu M$, both approximate the full model equally well (see Appendix \ref{sec:accuracy}).

Recall that we can also derive the Rational ERK model via singular perturbation theory.
When using perturbation, it is uncommon to bound the approximation error as explicitly as we do via the algebraic methods of \cite{reduction}. 
However, we can still show that the Linear ERK model is a good approximation of the Rational ERK model when $0\leq\gamma_1,\gamma_2\ll1$. 
Again, we can use knowledge of the Michaelis-Menten constant to show that in our experimental setup, $\gamma_1$ and $\gamma_2$ are small.
Indeed, we can rewrite
$$\gamma_1=\frac{1}{k_{M_1}}\frac{k_{c_1}+k_{c_2}}{ k_{c_2}+k_{r_2}}, \qquad \gamma_2=\frac{1}{ k_{M_2}}\frac{k_{c_1}+k_{c_2}}{ k_{c_1}+k_{r_1}}.$$
Since $k_{M_i}\approx 25\mu M$ and the parameters $k_{c_i}$ and $k_{r_i}$ are of similar magnitude 
(see \cite{BarEven11}),  
we conclude that $\gamma_1\approx 1/25\;(1/\mu M)$.

We reiterate that by employing an algebraic approach, we can derive a reduced model (without taking further limits) that approximates the Full ERK model as well as that obtained via singular perturbation theory, but has several advantages: it has fewer parameters, is interpretable as a chemical reaction network, and  identifiable, as discussed in the next Section.

\subsubsection{Choice of Output Variables}\label{sec:output_var}

Recall from Subsection \ref{sec:data}, the experimental measurements correspond to
the following linear combinations of variables: $S_0+C_1$, $S_1+C_2$, and $S_2$. Here we argue that in the context of available data, $S_0$, $S_1$, and $S_2$ are sufficient approximations of the output variables, which simplifies both the identifiability analysis and the parameter inference.

We argued in Subsection \ref{subsec:qssa-accuracy} that  in the context of experimental data   
the Linear ERK model is as good of an approximation to the Full ERK model as the Rational ERK model. On the long timescale, substitutions for $C_1$ and $C_2$ from the Linear ERK model give approximately
\begin{flalign*}
C_1&=\frac{1}{k_{M_1}}E_{tot}\cdot S_0, \\
C_2&=\frac{1}{k_{M_2}}E_{tot}\cdot S_1 + \frac{k_{c_1} }{ k_{c_2}+k_{r_1}}\frac{1}{k_{M_2}}E_{tot}\cdot S_0.
\end{flalign*}
Recall that $k_{M_i}\approx 25\mu M$ and $E_{tot}=0.65\mu M$.
We then find that the measurements of $S_i + C_{i+1}$ will be dominated by $S_i$. 
Henceforth we will use $S_i$ interchangeably with our measurements $S_i+C_{i+1}$.

%------------------------------------------------
%-------------------------------------------------

\section{Identifiability}\label{sec:ident}
One of the goals of this ERK study is to determine the kinetic parameters of the models given the data. Each model and experimental setup induces a map from the space of model parameters to observable model solutions (here, this is the measurement of the 3 species at the 7 time points over the course of $r$ experimental replicates, i.e., a subset of $\RR^{21r}$). We call this map $\phi_{t_1,\ldots,t_7}\colon\Theta\to \RR^{21r}$ the \emph{model prediction map} (see \cite{Dufresne2018}). Here, the parameter space $\Theta$ is a subset of the positive octant $\RR_{\geq 0}^6$ for the Full ERK model, $\RR_{\geq 0}^5$ for the Rational ERK model, and $\RR_{\geq 0}^3$ for the Linear ERK model. One can think of the data as being a point $z^*$ in the space of observable model solutions, i.e., $\RR^{21r}$, and parameter estimation corresponds to attempting to compute the inverse image $\phi_{t_1,\ldots,t_7}^{-1}(z^*)$ of this map at that point. \textit{Structural identifiability} generally corresponds to the model prediction map $\phi_{t_1,\ldots,t_7}$ being injective. Real-world observations are noisy, hence the data point $z^*$ may not be in the image of the map $\phi_{t_1,\ldots,t_7}$. Thus, when performing parameter estimation, we instead search for parameters yielding model predictions close to the data point $z^*$. \textit{Practical identifiability} broadly corresponds to having the set of parameters with model predictions close to the data point $z^*$ being bounded. In Subsection \ref{subsec:SI} we show that the Linear ERK model is structurally identifiable on its whole parameter space, while the Rational ERK model and the Full ERK model are structurally identifiable on some open dense subset of their parameter space. In Subsection \ref{sec:prac_id} we show that the Linear ERK model is practically identifiable for our experimental data, providing the proof of Theorem \ref{thm:LinearERKIdent}. By contrast, we provide evidence that the Rational ERK model and Full ERK model are not practically identifiable.

\subsection{Structural Identifiability}\label{subsec:SI}

First, we study the structural identifiability of our ODE models, that is whether the model prediction map $\phi_{t_1,\ldots,t_7}\colon\Theta\to \RR^{21r}$ is one-to-one, or at least locally one-to-one. We start by providing a formal definition of structural identifiability for models given by ODE systems with specific time points. Suppose we have a rational ODE system in variables $x\in\RR^n$ and parameters $\theta\in\RR^m$, given by
\begin{equation}
\frac{\mathrm{d} x}{\mathrm{d}t} = f(x,\theta), \label{eq:odeRat}
\end{equation} 
where $f$ is a vector of rational functions in $\RR(x,\theta)^n$. We assume that the measurable output is $y=g(x,\theta)$ where $g$ is also a vector of rational functions. Let $\hat{x}(\theta,t)$ be a solution of (\ref{eq:odeRat}) for the parameter value $\theta\in\Theta$ and then let $\hat{y}(\theta, t)=g(\hat{x}(\theta,t),\theta)$ be the observable solution for the same parameter value. Then, supposing that there are $r$ replicates of the experiment, for the specific time points $t_1,\ldots,t_l$ the model prediction map is given by \[\phi_{t_1,\ldots,t_l}(\theta)=\underbrace{(\hat{y}(\theta,t_1),\ldots,\hat{y}(\theta,t_l),\dots,\hat{y}(\theta,t_1),\ldots,\hat{y}(\theta,t_l))}_{r \text{ times}}.\] The model prediction map then induces an equivalence relation $\sim_{t_1,\ldots,t_l}$ on the parameter space $\Theta$ via 
\[\theta \sim_{t_1,\ldots,t_l}\theta' \text{ if and only if } \phi_{t_1,\ldots, t_l}(\theta)=\phi_{t_1,\ldots, t_l}(\theta'),\]
for any $\theta,\theta'\in\Theta$. 

\begin{definition}[c.f. Definition 2.8 in \cite{Dufresne2018}]
Suppose we have a model given by a system of rational ODEs (as above) with parameter space $\Theta$ and model prediction map $\phi_{t_1,\ldots,t_l}$. We say a model is:
\begin{itemize}
    \item  \emph{globally identifiable} if every equivalence class of $\sim_{t_1,\ldots,t_l}$ on $\Theta$ has size exactly 1.
    \item  \emph{generically identifiable} if for almost all $\theta\in\Theta$ the equivalence class of $\theta$ has size exactly 1.
    \item  \emph{locally identifiable} if for almost all $\theta\in\Theta$ the equivalence class of $\theta$ is finite.
    \item  \emph{generically non-identifiable} if for almost all $\theta\in\Theta$ the equivalence class of $\theta$ is infinite.
\end{itemize}
Here ``almost all'' means everywhere except possibly in a closed subvariety (i.e. the set of common zeroes of some polynomials).
\end{definition}

There are several approaches to assess  structural identifiability. All identifiability methods involve a certain number of assumptions of genericity, but not always explicitly (see for example discussions in \cite{Ovchinikov2021,Hong2019,Joubert2021,Saccomani2003,VILLAVERDE2018,Villaverde2019}). First, all methods assume that one has access to the whole trajectory of the observable output, and so are looking at the size of the equivalence classes of the equivalence relation $\sim_\infty$ on $\Theta$ defined as 
\[\theta\sim_\infty\theta' \text{ if and only if } \hat{y}(\theta,t)= \hat{y}(\theta',t) \text{ for all } t\geq 0.\]
 For rational ODE models with time series data as considered here, a result of Sontag \cite{sontag2002differential} proves if at least $2m+1$ generic time points are observed, where $m$ is the dimension of the parameter space, then the equivalence relation $\sim_{t_1,\ldots,t_{2m+1}}$ coincides with the equivalence relation $\sim_\infty$. If there are fewer time points or they are not generic, it could be that 
 almost all equivalence classes of $\sim_\infty$ have size 1 but those of $\sim_{t_1,\ldots,t_l}$ are larger. For the Linear ERK model, the parameter space has dimension 3, so we have enough time points, although we do not know a priori if they are generic. In fact, this model admits analytic solutions (See Subsection \ref{sec:prac_id}), so we can build the model prediction map explicitly and determine its identifiability directly. By a straightforward computation, we can show that for any choice of three distinct non-zero time points, the model prediction map $\phi_{t_1,t_2,t_3}$ of the Linear ERK model is injective and so the model is globally structurally identifiable (see Appendix \ref{ap:LinearSIdirect} for details). In particular, it follows that any choice of three distinct time points is generic. For the Rational ERK model and the Full ERK model, the parameter space has dimensions 5 and 6, respectively, hence we may not have enough time points, and we cannot determine the validity of any structural identifiability results for these specific model prediction maps. Indeed, these two models are non-linear and do not admit analytic solutions that would allow us to make the same argument as for the Linear ERK model. This is an instance of a more general open problem:
 
 \begin{problem}
 Find a criterion to determine structural identifiability for specific time points.
 \end{problem}

Methods to assess the structural identifiability of ODE models include the classical approach via Taylor series \cite{Pohjanpalo1978} and generating series \cite{grewal1976identifiability}, and, more recently, approaches based on differential algebra \cite{Audoly2001, Saccomani2003, Hong2019}. 
In this paper, we use SIAN \cite{SIAN}, an approach based on differential algebra implemented in Maple \cite{maple}.

Similar to other methods based on differential algebra (for example, the method implemented in DAISY \cite{DAISY}), SIAN is based on the differential Nullstellensatz \cite[Chapter 1]{Ritt-book} or \cite[Section 4]{seidenberg52}.
For a differentially closed field $\mathbb{K}$, this theorem establishes a correspondence between radical differential ideals and differentially closed subsets of $\mathbb{K}^n$. In the context of an ODE system, this implies that the solutions of the ODE system are completely determined by a prime differential ideal in a differential ring (see below). Criteria for identifiability can then be extracted from the ideal (or the quotient ring). The requirement that $\mathbb{K}$ is differentially closed then means that the solutions in question are possibly complex-valued, and the identifiability results will be about complex parameters, whether this is stated explicitly or not. For this reason, Hong et al. \cite{Hong2019} state their definition for complex parameters.

\begin{remark*}
As mentioned above, the first difference between our definition of identifiability and Hong et al.'s is that their parameter space is a subset of $\CC^n$ instead of $\RR^n$. 
A second difference to note is that what Hong et al. \cite{Hong2019} call ``globally identifiable'' corresponds to what we call generically identifiable. Finally, Hong et al.'s \cite{Hong2019} definition is written for components of the parameters and makes the notion of ``almost all'' more precise.
\end{remark*}

 The starting point is an ODE system of the same form as in (\ref{eq:odeRat}) together with the initial condition $x(0)=x_0$. Let $Q$ be the least common multiple of all the polynomials appearing in the denominators in $f$ and $g$, then we have $f=F/Q$ and $g=G/Q$ where $F$ and $G$ are polynomial functions. Note that SIAN usually views the initial conditions as additional unknown components of the parameter that one may want to identify.
The differential ring of interest is the differential ring $\CC(\theta)\{x,y\}$. We can think of this ring as a polynomial ring in infinitely many indeterminates: $\theta$, $x$, $y$, and the infinitely many higher derivatives of $x$ and $y$ (i.e., $x^{(i)}$ and $y^{(i)}$ for $i\geq 1$).  We are interested in  differential ideal $I_{\Sigma}$ of $\CC(\theta)\{x,y\}$ given by
\begin{align}
    I_{\Sigma}:=((Q\dot{x_i}-F_i)^{(j)},(Q\dot{y_k}-G_k)^{(j)}\mid 1\leq i\leq n,~\mid 1\leq i\leq m,j\geq 0):Q^\infty,
\end{align}
where for non-empty subsets $T,S$ of a ring $R$, the set $T:S^\infty$ is defined as follows: $$T:S^\infty:=\{r\in R\mid \text{there exist } s\in S, \,n\in\ZZ_{\geq 0} \text{ such that } s^nr\in T\}.$$
Note that for polynomial systems like the Full ERK model and the Linear ERK model, we have $Q=1$, and so the column operation is not needed and the ideal $I_{\Sigma}$ is simply the differential ideal generated by the equations defining the ODE system and their derivatives. The ideal $I_{\Sigma}$ is the ideal of all differential polynomials in $\CC(\theta)\{x,y\}$ that vanish on the solutions of the system of ODE system (\ref{eq:odeRat})  \cite{Saccomani2003, Hong2019}.

The ideal $I_{\Sigma}$ is prime \cite{Hong2019} and so the quotient ring $\CC(\theta)\{x,y\}/I_{\Sigma}$ is an integral domain. Let $\mathbb{K}:=Q(\CC[\theta]\{x,y\}/I_{\Sigma})$ be the field of fractions of the domain $\CC(\theta)\{x,y\}/I_{\Sigma}$, and let $\Bbbk$ be the subfield of $\mathbb{K}$ generated by the image of $\CC\{y\}$, that is, the subfield generated by the elements of the form $y_i+I_{\Sigma}$. We can now state the non-constructive algebraic criterion for structural identifiability:

\begin{proposition}[c.f. Proposition 3.4 in \cite{Hong2019}]
Suppose we have a model given by a system of rational ODEs as described above.
\begin{itemize}
    \item If the fields $\Bbbk$ and $\Bbbk(\theta)$ coincide, then the model is generically identifiable. 
    \item If the field extension $\Bbbk\subseteq \Bbbk(\theta)$ is algebraic, then the model is locally identifiable.
\end{itemize}
\end{proposition}

\begin{remark*}
Note that Proposition 3.4 in \cite{Hong2019} implies that $\Bbbk$ and $\Bbbk(\theta)$ coincide (respectively the field extension $\Bbbk\subseteq \Bbbk(\theta)$ is algebraic) if and only if the model is globally identifiable (respectively, locally identifiable) in the sense of \cite{Hong2019}. We are interested in something weaker, we only wish to identify parameters in the parameter space $\Theta$, which is a subset of the real positive octant.
\end{remark*}

The criterion provided by the proposition above is not constructive, as it involves the field of rational functions of an infinitely generated $\CC$-algebra. Hong et al. \cite{Hong2019} go on to provide a constructive version of the criterion \cite[Section 3]{Hong2019}. The software SIAN \cite{SIAN}, which we use here, is in turn based on a probabilistic version of the criterion \cite[Section 4]{Hong2019}. Note that local identifiability is determined via the Taylor series approach.

We now consider the issue of initial conditions. As mentioned above, by default, SIAN considers the initial conditions as parameters that one may wish to identify. Other methods, like the differential algebra method as implemented in DAISY \cite{DAISY}, do not explicitly address initial conditions.  Ovchinnikov et al. show in  \cite[Theorem 19]{IdentifiableFunctions} that input-output identifiability corresponds to what they call multiple experiment identifiability, that is, identifiability from sufficiently many generic initial conditions. DAISY and COMBOS verify input-output identifiability \cite{meshkat2014finding}.

Using SIAN \cite{SIAN}, we verify that all three models are generically identifiable. Recall that this result is valid under the assumption that we have measurements at sufficiently many generic time-points, and for generic initial conditions. Inspired by the discussion in \cite{Saccomani2003}, in Appendix \ref{an:InitialConditions} we show that the set of differential polynomials in $\CC(\theta)\{x,y\}$ vanishing on those solutions of the system \ref{eq:odeRat} with initial conditions $S_0(0)=5\mu M$ and $S_1(0)=S_2(0)=0\mu M$ for all three models, as well as $C_1(0)=C_2(0)=0\mu M$ and $E(0)=0.65\mu M$ for the Full ERK model coincides with the ideal $I_\Sigma$. This means that the set of solutions with initial conditions corresponding to our experimental setup is dense in the set of all solutions for the Kolchin topology (induced by the differential ideals of $\CC(\theta)\{x,y\}$. We can therefore conclude that the initial conditions specific to the experimental setup are indeed generic.

\begin{remark*}
Using SIAN we can show that the Full ERK model is also generically identifiable with measurable outputs $S_0+C_1$, $S_1+C_2$ and $S_2$ which is what was actually measured experimentally (see Subsection ~\ref{sec:output_var}).
\end{remark*}

%--------------------------------------

\subsection{Practical Identifiability}\label{sec:prac_id}

Suppose a model is generically identifiable, then, generically, distinct parameters produce distinct data points. However, if there are parameter values that are arbitrarily far from one another but produce data points close to each other, parameter estimation would not be meaningful in practice. 
Practical (non-)identifiability aims to categorise models exhibiting such undesirable behaviour.
For example, sloppiness \cite{gutenkunst2007universally}, uncertainty quantification \cite{smith2013uncertainty} and filtering problems \cite{shi1999kalman} study mathematical models with a similar aim.
We use a definition of practical identifiability introduced in \cite{Dufresne2018}, which was adapted from the definition given in \cite{raue2009structural}.

Practical identifiability depends on more than the defining equations and specification of input and output of the model. Practical identifiability will be influenced by the precise choice of time points, the method used for parameter estimation, the assumption on measurement noise of the data, and the way we measure distances in parameter space. It may also vary on the area in the data space. A data point $z^*$ is an experimental observation in the form of an $N$-dimensional vector whose entries are the observed values of the measured variables at each of the specific time points for each replicate of the experiment. We focus on practical identifiability for maximum likelihood estimation (MLE), one of the most widely used methods for parameter estimation (see, for example \cite{ljung1987theory}). Accordingly, in the remainder of this section, we consider models $(\mathcal{M},\phi_{t_1,\ldots,t_s},\psi,d_\Theta)$ with a precise choice of model prediction map $\phi_{t_1,\ldots,t_s}$ with specific time points ${t_1,\ldots,t_s}$, a specific assumption for the probability distribution $\psi$ of measurement noise and a choice of reference metric $d_\Theta$ on parameter space $\Theta$. We will also assume that the model considered is at least generically identifiable, so that MLE exist and are unique for generic data (see \cite[Proposition 4.15]{Dufresne2018}). 
We write $\hat{\theta}(z^*)$ to denote the MLE for $z^*$, that is, $\hat{\theta}(z^*):=\operatorname{max}_{\theta\in \Theta}\psi(\theta,z^*)$.

We define an \emph{$\delta$-confidence region ${U}_\delta(z^*)$} as follows:
\begin{align*}
U_{\delta}(z^*) := \{\theta \in \Theta \mid - \log \psi(\theta,z^*) < \delta \}.
\end{align*}
The set $U_{\delta}(z^*)$, often known as a likelihood-based confidence region \cite{Vajda1989, gc-rlbstatsBook}, is intimately connected with the likelihood ratio test. Specifically, suppose we had a null hypothesis $\mathbf{H}_0$ that data point $z^*$ has true parameter $\theta^*$, and we wished to test the alternative hypothesis $\mathbf{H}_1$ that $z^*$'s true parameter is something else. By definition, a likelihood ratio test would reject the null hypothesis when
\begin{align*}
\Lambda(\theta^*,z^*):= \frac{\psi(\theta^*,z^*)}{\psi(\hat{\theta}(z^*),z^*)} \leq k^*,
\end{align*}
where $k^*$ is a critical value, with the significance level $\alpha$ equal to the probability $\text{Pr}(\Lambda(z^*)\leq k^* | \mathbf{H}_0)$ of rejecting the null hypothesis when it is in fact true. The set of parameters such that the null hypothesis is not rejected at significance level $\alpha$ is
\[ \{ \theta' \in \Theta \mid -\log\psi(\theta',z^*)<-\log\psi(\hat{\theta}(z^*),z^*)-\log k^*\},\]
that is, $U_\delta(z^*)$, where $\delta=-\log\psi(\hat{\theta}(z^*),z^*)-\log k^*$.

\glsxtrnewsymbol[description={A $\delta$-confidence region}]{Ud}{\ensuremath{U_\delta}}
\glsxtrnewsymbol[description={Likelihood ratio}]{Lambda}{\ensuremath{\Lambda}}

\begin{definition}[{\cite[Definiton 4.17]{Dufresne2018}}]
The model $(M,\phi_{t_1,\ldots,t_s},\psi,d_{\Theta})$ is practically identifiable for a data point $z^*\in \RR^N$ at significance level $\alpha$ if and only if the confidence region $U_{\delta}(z^*)$ is bounded with respect to  $d_\Theta$, where 
\[\delta=-\log \psi(\hat{\theta}(z^*),z^*)-\log k^*\]
and
\begin{equation}
\alpha=\operatorname{Pr}\left(\frac{\psi(\hat{\theta}(z^*),\hat{z})}{\operatorname{max}_{\theta\in \Theta}\psi(\theta,\hat{z})}<k^* \mid \hat{z} \text{ is data with true parameter }\hat{\theta}(z^*)\right).\label{eq:k_star}
\end{equation}
\end{definition}

\glsxtrnewsymbol[description={A significance level}]{alpha}{\ensuremath{\alpha}}

\glsxtrnewsymbol[description={Number of replicates}]{r}{\ensuremath{r}}
\glsxtrnewsymbol[description={Minimal log-likelihood of a point in parameter space to be be included in $\alpha$-confidence region}]{delta}{\ensuremath{\delta}}
\glsxtrnewsymbol[description={Standard deviation of a distribution}]{sigma}{\ensuremath{\sigma}}
For our analysis, we make the common assumption that the measurement noise is additive Gaussian with covariance matrix equal to a multiple of the identity matrix. The assumption is implicit when performing a least-squares fit computation for MLE. In our setup, we are measuring 3 substances at 7 time-points and there were $r$  replicates, so our assumption on the measurement noise means that the probability distribution of the data is given by
\[\psi(\theta,z)=(2\pi \sigma^2)^{\frac{-21r}{2}}e^{-\frac{1}{2\sigma^2}\|z-\phi_{t_1,\ldots,t_7}(\theta)\|_2^2},\]
where $\sigma^2 I_{21}$ is the covariance. It then follows that 
\begin{align*}
\delta&=-\log \psi(\hat{\theta}(z^*),z^*)-\log k^*\\
         &=\frac{21r}{2}\log(2\pi \sigma^2)+\frac{1}{2\sigma^2}\|z^*-\phi_{t_1,\ldots,t_7}(\hat{\theta}(z^*))\|_2^2-\log k^*,
\end{align*}
and
\begin{align*}
-\log \psi(\theta',z^*)=\frac{21r}{2}\log(2\pi \sigma^2)+\frac{1}{2\sigma^2}\|z^*-\phi_{t_1,\ldots,t_7}(\theta')\|_2^2.
\end{align*}
Therefore, we have that
\begin{align*}
U_\delta(z^*)&=\{\theta'\in\Theta\mid  \|z^*-\phi_{t_1,\ldots,t_7}(\theta')\|_2^2 <\|z^*-\phi_{t_1,\ldots,t_7}(\hat{\theta}(z^*))\|_2^2-2\sigma^2\log k^*\}\\
                      &=\phi_{t_1,\ldots,t_7}^{-1}({B}_{\rho}(z^*)),
\end{align*}
where ${B}_{\rho}(z^*)$ is the Euclidean open ball of radius $\rho:=\sqrt{(\|z^*-\phi_{t_1,\ldots,t_7}(\hat{\theta}(z^*))\|_2^2-2\sigma^2\log k^*)}$ around the data point $z^*$. It follows that under our assumptions, determining whether the various models we study are practically identifiable corresponds to determining whether the preimages under the model prediction map of small open balls around data points are bounded in parameter space. The size of the balls will depend on the data point and the significance level $\alpha$ (or equivalently the critical value $k^*$).

%--------------------------------------------------------------------------------------------------

\subsection{The practical identifiability of the Rational ERK model and the Full ERK model} \label{sec:PractIdentRat}

The Rational ERK model and the Full ERK model do not admit analytic solutions, hence we do not have access to an explicit model prediction map $\phi_{t_1,\ldots,t_l}$. Therefore, we must approximate $\phi_{t_1,\ldots,t_l}$ and thus also $U_\delta$ using numerical methods and repeated sampling.

First, we assume that our measurements have been corrupted with some Gaussian noise with mean 0 and variance $\sigma^2$. This variance is identical across measurement quantities, time points, and trials. The noise distributions are independent across measurements.

As we have assumed that measurement noise is additive Gaussian with covariance matrix equal to a multiple of the identity matrix, we can obtain an MLE, given some data $z^*$, by solving a least squares problem. This gives us $\hat{\theta}(z^*)$. We use this parameter to calculate the sample variance, assuming that the mean of each quantity is the model trajectory at each time point. This gives us an estimate of the covariance $\sigma^2$.

Recall that $\delta$ is defined to be $-\log \psi(\hat{\theta}(z^*),z^*)-\log k^*$. The log-likelihood is easy to compute, as we already know $z^*$ and $\hat{\theta}(z^*)$, and can estimate $\phi_{t_1,\ldots, t_7}$ using a numerical solution to the ODE system.  We use the following procedure to approximate $-\log k^*$:

\vspace{.05in}

\begin{algorithm}[H] \label{algo:-logk*}
 \caption{\texttt{Computing approximate $-\log k^*$}}
 \KwData{$z^*$, $\hat{p}(z^*)$, $\sigma^2$, $\alpha$, $n_{\mathrm{iter}}$ (number of iterations)}
 $\{z'_i\}_{i=1,..,n_{\mathrm{iter}}} \leftarrow$ $n_{\mathrm{iter}}$ corruptions of $z^*$ by adding i.i.d. random samples from $\mathcal N(0,\sigma^2)$ to each measurement\\
 $LogLik\leftarrow[]$ (create an empty list)
 
 \For{$i=1,...,n_{\mathrm{iter}}$}{
 	Calculate $\hat{p}(z'_i)$ by least-squares solving\\
 	Append $\log\psi(\hat{p}(z'_i), z'_i)-\log\psi(\hat{p}(z^*), z'_i)$ to $LogLik$
 }
 \Return $\lfloor n_{\mathrm{iter}} \cdot(1-\alpha)\rfloor$-th largest element of $LogLik$
 
\end{algorithm}

\vspace{.05in}

\noindent This simply follows the definition of $k^*$ in Equation (\ref{eq:k_star}), and approximates $-\log k^*$ by repeatedly sampling likelihood-ratios under our given noise assumptions and then taking a $(1-\alpha)$-quantile (as $-\log(\,\cdot\,)$ is a monotonically decreasing function).

\begin{remark*}
In a situation where the number of replicates $r$ is large, an approximate $\delta$ can be computed from $\alpha$ that depends primarily on the distance between the data point $z^*$ and the predicted data point $\phi_{t_1,\ldots,t_7}(\hat{\theta}({z^*}))$ corresponding to the MLE.

From the definition, we have $\delta=-\log \psi(\hat{p}(z^*),z^*)-\log k^*$, meaning that $k^*=1/e^{\delta}\psi(\hat{p}(z^*),z^*)$, and so we can describe $\alpha$ in terms of $\delta$ directly:
\[\alpha=\operatorname{Pr}\left(\frac{\psi(\hat{\theta}(z^*),\hat{z})}{\operatorname{max}_{\theta\in\Theta}\psi(\theta,\hat{z})}<1/e^{\delta}\psi(\hat{\theta}(z^*),z^*) \mid \hat{z} \text{ is data with true parameter }\hat{\theta}(z^*)\right).\] 
This is equivalent to
\[\alpha=\operatorname{Pr}\left(-\log\left(\frac{\psi(\hat{\theta}(z^*),\hat{z})}{\operatorname{max}_{\theta\in \Theta}\psi(\theta,\hat{z})}\right)>-\log(1/e^{\delta}\psi(\hat{\theta}(z^*),z^*) \mid \hat{z} \text{ has true parameter }\hat{\theta}(z^*)\right),\]
and so
\[\alpha=\operatorname{Pr}\left(-\log \psi(\hat{\theta}(z^*),\hat{z})+\log{\operatorname{max}_{\theta\in\Theta}\psi(\theta,\hat{z})}>\delta+\log \psi(\hat{\theta}(z^*),z^*) \mid \hat{z} \text{ has true parameter }\hat{\theta}(z^*)\right).\]
Note that for each value of $\hat{z}$, the MLE $\hat{\theta}(\hat{z})$ maximises $\psi(\theta,\hat{z})$. It follows that
\[\alpha=\operatorname{Pr}\left(2(\log \psi(\hat{\theta}(\hat{z}),\hat{z})-\log \psi(\hat{\theta}(z^*),\hat{z}))>2\delta+2\log \psi(\hat{\theta}(z^*),z^*) \mid \hat{z} \text{ has true parameter }\hat{\theta}(z^*)\right).\]
Wilk's theorem \cite{Wilks} implies that $2(\log \psi(\hat{\theta}(\hat{z}),\hat{z})-\log \psi(\hat{\theta}(z^*),\hat{z}))$ is asymptotically $\chi^2$ with three degrees of freedom. 
 If $F(\hat{z})$ is the asymptotic cumulative distribution function of $2(\log \psi(\hat{\theta}z),\hat{z})-\log \psi(\hat{\theta}(z^*),\hat{z}))$, then $\alpha$ is approximately equal to
\begin{align*}
       \alpha     &=1-\operatorname{Pr}\left(2(\log \psi(\hat{\theta}(\hat{z}),\hat{z})-\log \psi(\hat{\theta}(z^*),\hat{z}))<2\delta+2\log \psi(\hat{\theta}(z^*),z^*) \mid \hat{z} \text{ has true param. }\hat{\theta}(z^*)\right)\\
            &\approx 1-F(2\delta+2\log \psi(\hat{\theta}(z^*),z^*)).\end{align*}
 Therefore, asymptotically we have that
 \[\delta=F^{-1}(1-\alpha)/2-\log \psi(\hat{\theta}(z^*),z^*).\]
Unfortunately, this is not applicable here, as the number of experiments here is 5, 6 or 11, which are not large numbers. Indeed, the $\delta$ obtained by applying Wilks' Theorem and the $\delta$ obtained via Algorithm \ref{algo:-logk*} are notably different. For example, for the wild-type and the Linear model, we approximate $-\log k^*$ as $0.477$ while Wilks' theorem approximates it as $3.907$.
\end{remark*}

In order to demonstrate practical non-identifiability for the Full and Rational ERK models, we pick two parameters from each model, based on which we can illustrate non-identifiability well by presenting confidence areas marginalised to these two parameters. This choice of parameters is informed by performing a (ill-posed) Bayesian parameter inference first (see next section). This procedure is described here for the Rational ERK model, but works similarly for the full model:

\vspace{.05in}

\begin{algorithm}[H]\label{alg:mca}
 \caption{\texttt{Computing marginalised confidence regions}}
 \KwData{$\delta$, $(\kappa_1^\mathrm{min}, \kappa_1^\mathrm{max})$, $(\gamma_1^\mathrm{min}, \gamma_1^\mathrm{max})$ (bounds on the parameters), $n_{\mathrm{iter}}$ (number of evaluations per parameter)}
 $CA\leftarrow[]$ (create an empty list)\\
 $\Delta\kappa_1 \leftarrow(\kappa_1^\mathrm{max}-\kappa_1^\mathrm{min})/n_{\mathrm{iter}}$\\
 $\Delta\gamma_1 \leftarrow(\gamma_1^\mathrm{max}-\gamma_1^\mathrm{min})/n_{\mathrm{iter}}$\\
 \For{$i=0,...,n_{\mathrm{iter}}-1$}{
 	$\hat{\kappa}_1 \leftarrow \kappa_1^\mathrm{min} + i\cdot\Delta\kappa_1$\\
 	\For{$j=0,...,n_{\mathrm{iter}}-1$}{
 		$\hat{\gamma}_1 \leftarrow \gamma_1^\mathrm{min} + j\cdot\Delta\gamma_1$\\
 		Find parameters $(\kappa_2',\pi', \gamma_2')$ minimising $\psi((\hat{\kappa}_1, \kappa_2',\pi',\hat{\gamma}_1,\gamma_2'),z^*)$ by least-squares solving\\
 		\If{$-\log\psi((\hat{\kappa}_1, \kappa_2',\pi',\hat{\gamma}_1,\gamma_2'),z^*)<\delta$}{Append $(\hat{\kappa}_1, \hat{\gamma}_1)$ to $CA$}
 	}
 }
 \Return $CA$
 
\end{algorithm}

\vspace{0.05in}

While we do not know the values of $\kappa_1$ and $\gamma_1$, previous experimental work has provided bounds for $\kappa_1$ and $\gamma_1$, which we pass to the algorithm above. The list returned by the algorithm is a discrete approximation of the confidence area, marginalised to the pair of parameters $\kappa_1$ and $\gamma_1$. We plot these points for visual inspections, which can be seen in Figure \ref{fig:qssa_ca}. The blue area reaching the upper and leftmost boundary of the plot indicates that the confidence region is very unlikely to be bounded and that this model is very unlikely to be practically identifiable.
\begin{figure}
    \centering
    \includegraphics[width=0.8\textwidth]{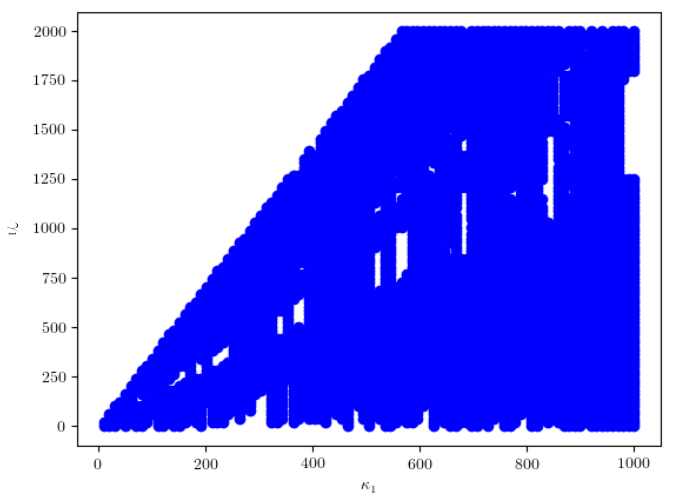}
    \caption{Marginalised confidence area following Algorithm \ref{alg:mca} at significance level 0.05  for the Rational ERK model for the wild-type data point  $z^*$, with $\kappa_1^{\mathrm{min}}=0\,(1/min)$, $\kappa_1^{\mathrm{max}}=1000\,(1/min)$, $\gamma_1^{\mathrm{min}}=0\,(1/\mu M)$, and $\gamma_1^{\mathrm{max}}=1000\,(1/\mu M)$.}
    \label{fig:qssa_ca}
\end{figure}

The source of this practical non-identifiability of the Full ERK model and the Rational ERK model is not completely clear. One possible source of non-identifiability could be the choice of time points. Indeed, as mentioned in Subsection \ref{subsec:SI}, in both cases we do not know if the time points are sufficiently generic. There are reasons to believe that not all practical non-identifiability can be explained by having an insufficient number of time points. Indeed, as part of earlier work during the preparation of \cite{Yeung2020}, additional time point data was simulated for the Full ERK model, but confidence regions to still appeared unbounded.  Another possible source of non-identifiability could be that for the given experimental data there is a valid quasi-steady-state approximation resulting in a smaller dimensional parameter space. At quasi-steady-state parameter values, the reduction is exact and so for these parameters, the equivalence class of $\sim_{t_1,\ldots,t_7}$ is positive dimensional. Intuitively, since the solutions of the Full ERK model and the Rational ERK model are close to those of the Linear ERK model near quasi-steady-state parameter values, the confidence regions should contain the equivalence class of the nearby quasi-steady-state parameter value, which in this case, was unbounded. This might be an example of more widespread phenomena.

%--------------------------------------------------------------------------------------------------

\subsection{The practical identifiability of the Linear ERK model}\label{subsec:PILinear}

We now consider the practical identifiability of the Linear ERK model. What distinguishes the Linear ERK model from the Full ERK model and the Rational ERK model is that an analytic solution to the ODE system is available and so we can construct an explicit model prediction map.
The solution to the ODE system (\ref{eq:l}) with
initial conditions $S_0(0)=5\mu M$ and $S_1(0)=S_2(0)=0$ is given by:

\begin{align*}
 {S}_0(t)= &5e^{-\kappa_1t}\\
{S}_1(t)= & 5\kappa_1(1-\pi)t e^{-\kappa_1t}                                                            & \text{if } \kappa_1=\kappa_2\\
                    & 5\kappa_1(1-\pi)(e^{-\kappa_2t}-e^{-\kappa_1t})/(\kappa_1-\kappa_2) & \text{otherwise }\\
{S}_2(t)=& 5-{S}_0(t)-{S}_1(t).
\end{align*}

As we did for the Rational ERK model in Subsection \ref{sec:PractIdentRat}, for a given data point $z^*$, we obtain an MLE $\hat{\theta}(z^*)$ by solving a least-squares problem. We then use Algorithm \ref{algo:-logk*} to approximate $-\log k^*$, and then $\delta$, using the explicit model prediction map we construct based on the analytic solutions. In Figure \ref{fig:potatoes} we plot the boundary of the confidence regions at significance level $\alpha=0.05$ for the data points corresponding to the wild-type and each mutant. All five confidence regions are seen to be bounded, and we conclude that the model is practically identifiable for those data points.

\begin{figure} 
    \centering
    \includegraphics[width=0.5\textwidth]{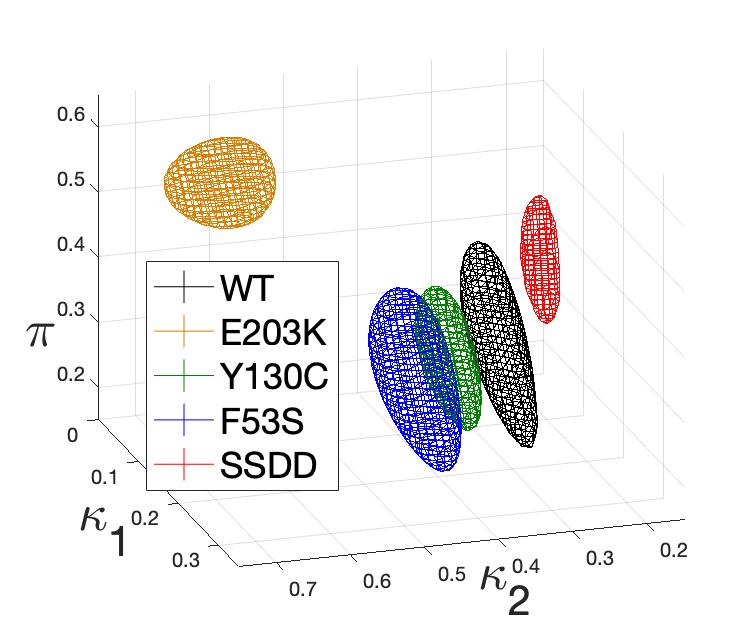}
    \caption{Boundary of the confidence regions for wild-type and each mutant at significance level 0.05 for the Linear ERK model.}
    \label{fig:potatoes}
\end{figure}

%.............................

%----------------------------------------
%----------------------------------------

%----------------------------------------------
%-----------------------------------------------

\section{Parameter Inference \& Topological Analysis}\label{sec:topology}

Having established that the Linear ERK model is practically identifiable, we now infer the parameters of this model using data from wild-type and mutant experiments.  First, we briefly review the Bayesian approach for inferring parameters of the Linear ERK model,
as already computed by Yeung et al \cite{Yeung2020}. 
We then introduce topological data analysis (TDA), analyse the point cloud of parameters sampled from the posteriors of the wild-type and four mutants, and compare their topological distances.

\subsection{Bayesian Inference}\label{sec:inference}

Given experimental data and a mathematical model, 
we seek to infer parameters for which the model accurately fits the data. We choose to do this via Bayesian inference. The theory of Bayesian statistics captures how our belief in the true value of these parameters changes when we make observations (in this case: measurements) in the language of probability theory. Most importantly, Bayesian inference does not infer a single value for each parameter, as would  a frequentist approach; rather, it infers a probability distribution of parameter values expressing how strongly we believe a certain set of parameter values is correct.

Formally, we are given a parameter space $\Theta$ and observations $x$ from some sample space $\mathcal X$. Combining the mathematical model with noise assumptions on available measurements, we obtain an expression for $p(x|\theta)$, the \emph{likelihood} of observing $x$ assuming that the parameter of the model is $\theta\in\Theta$. In addition, we need to specify a measure of belief in the parameter values before we observe any data, expressed through a probability density $p(\theta)$, called the \emph{prior distribution}. Theoretically, we want to inform a Bayesian inference only through observations. Consequently, we do not want to inform the inference by placing strong prior beliefs on certain parameter values. In practice, however, a trade-off between neutral prior beliefs (which should only account for substantive prior knowledge and possibly scientific conjectures), analytical convenience, and computational tractability is commonplace \cite[11-12]{gelman}.

Having selected a mathematical model and a prior distribution, our formal belief in parameter values becomes
$$p(\theta|x)\propto p(x|\theta)\cdot p(\theta)$$
by making observations $x\in\mathcal X$. The probability density $p(\theta|x)$ is called the \emph{posterior distribution}. The proportionality in the above equation indicates that we omitted a normalisation which is independent of $\theta$. As one can approximately sample from $p(\theta|x)$ without normalising, the normalisation factor is not necessary for our application.

For the Linear ERK model (Equations (\ref{eq:l})), the parameter is $\theta=(\kappa_1, \kappa_2, \pi, \sigma)\in\RR^4=\Theta$. Here, the first three components come from the parameter of the Linear ERK model while $\sigma$, the variance of the distribution of the data, which must be inferred in order to construct a Bayesian model, and will be subsequently marginalised (i.e., integrated out). The observations are measurements of $S_0$, $S_1$ and $S_2$. As measurements of each MEK type are taken from $r$  
replicates, at 7 different times, 
for 3 phosphorylation states of substrate, we formally have $\mathcal X=\RR^{r\cdot 3\cdot 7}=\RR^{r\cdot21}$. We have $r=11$ for the wild-type, $r=6$ for SSDD, and $r=5$ for all other variants. 

To construct a statistical model on the mechanistic Linear ERK model, we set the prior distributions to
$$\kappa_1,\kappa_2 \sim \text{\textit{Unif}}(0\,(1/\mathrm{min}),10\,(1/\mathrm{min})),\quad \sigma \sim \text{\textit{Unif}}(0\,(\mu M),10\,(\mu M)),$$
a uniform distribution over values we deem 
biologically feasible for these parameters \cite{Yeung2020}, and $\pi\sim \text{\textit{Unif}}(0,1)$, as $\pi$ can only take values within this range by definition.

Given samples $S_0^*$, $S_1^*$ and $S_2^*$, we assume that
\begin{align*}
\left(S_0^*\right)_{t,i}\sim\mathcal N\left(S_0(\kappa_1,\kappa_2,\pi,t),\sigma\right),\\
\left(S_1^*\right)_{t,i}\sim\mathcal N\left(S_1(\kappa_1,\kappa_2,\pi,t),\sigma\right),\\
\left(S_2^*\right)_{t,i}\sim\mathcal N\left(S_2(\kappa_1,\kappa_2,\pi,t),\sigma\right),
\end{align*}
where $t$ denotes the respective measurement time and $i$ indexes the sample. Here, $S_j(\kappa_1,\kappa_2,\pi,t)$ is a solution to the ODE system at time $t$ for parameters $\kappa_1$, $\kappa_2$ and $\pi$. For the Linear ERK model, we can construct an analytic solution to the governing equations, but generally, a numerical solution suffices. Such ODE solutions give rise to an expression for the likelihood $p(x\vert\theta)$.

\glsxtrnewsymbol[description={Parameter space}]{Theta}{\ensuremath{\Theta}}
\glsxtrnewsymbol[description={Data space (all measurements from a set of experiments)}]{Xcal}{\ensuremath{\mathcal{X}}}
\glsxtrnewsymbol[description={Measurement of species concentration $j$ at time $t$ in trial $i$}]{Sstar}{\ensuremath{(S_j^*)_{t,i}}}

We note that in the above Bayesian model, some standard simplifying assumptions were made. First, in the given setup, negative values of measurements of $S_0$, $S_1$ and $S_2$ have strictly positive likelihoods, which is not true in reality.
Second, we assume that $(S_0^*)_{t,i}$, $(S_1^*)_{t,i}$ and $(S_2^*)_{t,i}$ are independent random variables for all $t$ and $i$ and that they have the same standard deviation. Despite of these assumptions, we obtained good fits to the data. For example, performing an inference with three different standard deviation parameters $\sigma_0$, $\sigma_1$ and $\sigma_2$ for $S_0$, $S_1$ and $S_2$ respectively did not significantly improve the fits to the data.

This Bayesian inference framework can also be applied to other ODE models describing the measurements, including the Rational ERK model (Equations (\ref{eq:q})) and the Full ERK model (Equations (\ref{eq:a})). In these cases, we employ numerical solutions and adapt priors to the larger parameter spaces. 

We note that for the Full ERK model and the Rational ERK model, the choice of prior distributions significantly changes both the location and prominence of modes of the posterior distributions.  In particular, they tend to be near the endpoints of the prior distributions. This is linked to the practical non-identifiability of these models and prevents us from interpreting parameter modes, and also from conducting a sensible topological comparison that is not highly dependent on the choice of prior distribution.

In order to compute posterior distributions of the involved parameters, we used PyStan, the Python version of the statistical software STAN \cite{STAN}. While analytical expressions for the posterior distributions are too complex to be feasible for interpretation, PyStan enables us to approximately sample from them via Hamiltonian MCMC. The resulting samples (visualised in Figure \ref{fig:posterior}) form the basis of our further analysis.

\begin{figure}
\centering
\subfloat[]{
\includegraphics[width=0.4\textwidth]{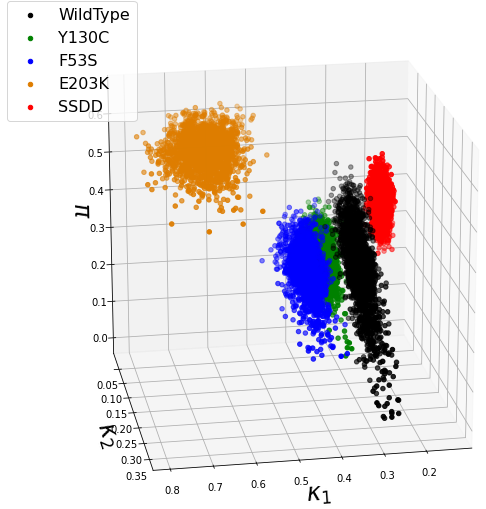}\label{fig:posterior_a}}%
\subfloat[]{
\includegraphics[width=0.4\textwidth]{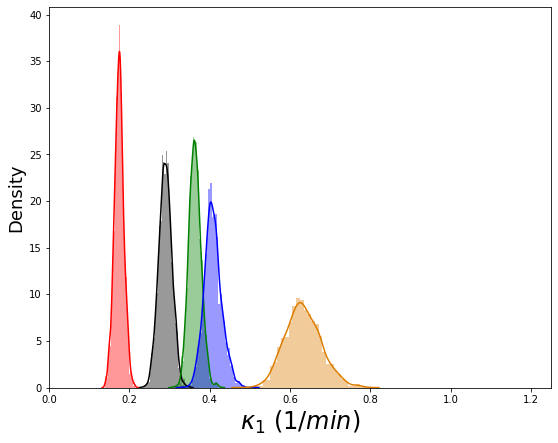}\label{fig:posterior_b}}
\\
\subfloat[]{
\includegraphics[width=0.4\textwidth]{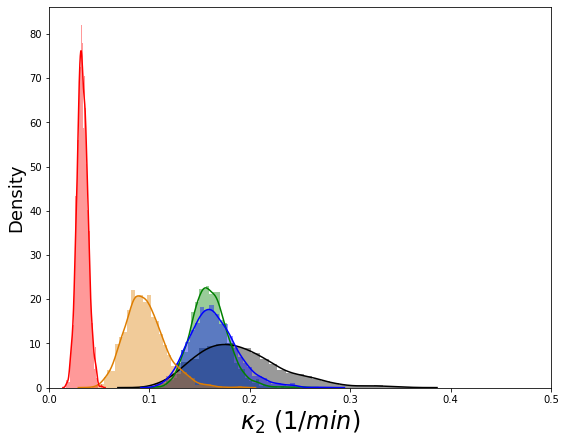}\label{fig:posterior_c}}%
\subfloat[]{
\includegraphics[width=0.4\textwidth]{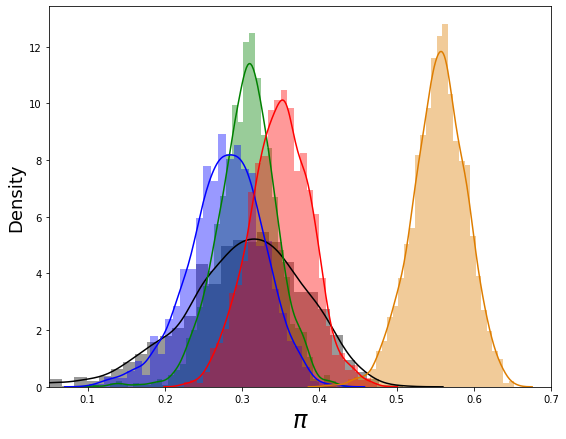}\label{fig:posterior_d}}
\caption{(\textsc{a}) Random  samples  from  the  posterior  distributions  for  the  WT  and  all mutants (2000 points each); Moreover, we display approximate marginal densities for $\kappa_1$ in (\textsc{b}), $\kappa_2$ in (\textsc{c}), and $\pi$ in (\textsc{d}) in the same colour scheme.}\label{fig:posterior}
\end{figure}

\subsection{Topological Analysis}

To analyse the topology of the samples of the resulting posterior distributions, we introduce notation and methodology from Topological Data Analysis (TDA). 

\begin{definition}
Let $\mathbf{v}$ be a finite set of vertices. A subset of the power-set of $\mathbf{v}$, $\KK\subseteq\mathcal P(\mathbf{v})$, is called a \emph{simplicial complex} if for any $\tau\in\KK$ the relation $\tau'\subseteq\tau$ implies $\tau'\in\KK$.

\glsxtrnewsymbol[description={A simplicial complex}]{K}{\ensuremath{\KK}}
\glsxtrnewsymbol[description={A simplex}]{tau}{\ensuremath{\tau}}
\glsxtrnewsymbol[description={A set of vertices}]{vcal}{\ensuremath{\mathbf{v}}}
\glsxtrnewsymbol[description={A simplicial map}]{h}{\ensuremath{h}}

We write $\KK_i=\{\tau\in\KK\,\vert\,\#\tau=i+1\}$ and call the elements of $\KK_i$ the $i$-simplices. A map $h:\mathbf{v}\to\mathbf{v}'$ which extends to a map $h:\mathcal{K}\to\mathcal{K}'$ by $h(\tau):=\{h(v)\,\vert\, v\in\tau\}$ for each $\tau\in\KK$ is called a \emph{simplicial map}. 
\end{definition}

\begin{figure}
\centering
\subfloat[]{
\includegraphics[width=0.65\textwidth]{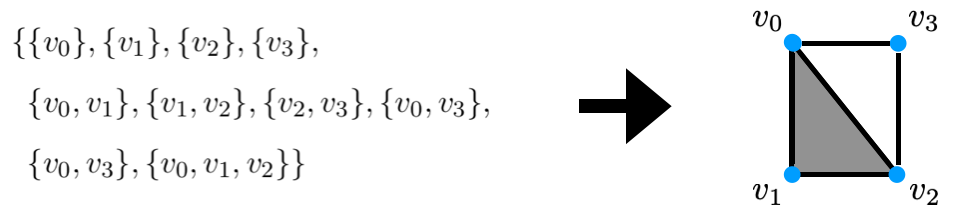}\label{fig:sc_a}}\\
\subfloat[]{
\includegraphics[width=0.7\textwidth]{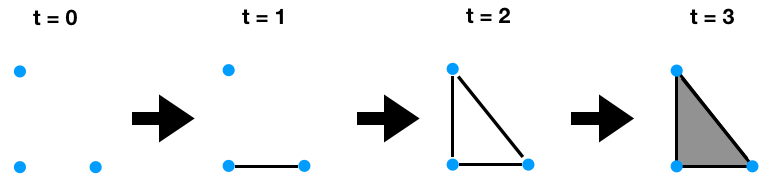}\label{fig:sc_b}}
\caption{(\textsc{a}) An example of a simplicial complex on vertices $v_0$, $v_1$, $v_2$, and $v_3$ (left) and its geometrical realisation (right); (\textsc{b}) An example of a filtration of a simplicial complex, visualised on geometric realisations.}
\end{figure}

We can view a simplicial complex as a combinatorial description of a topological space. Given a simplicial complex $\KK$, we can investigate its geometric realisation
$$\left|\KK\right|:=\bigcup_{\tau\in\KK}\mathrm{cvx}(\tau)\subseteq\RR\langle \mathbf{v}\rangle,$$
where $\mathrm{cvx}$ denotes the convex hull in the real free vector space generated by the vertices $V$. The realisation $|\KK|$ is endowed with the subspace topology in $\RR\langle \mathbf{v}\rangle$. An example of a simplicial complex and its geometric realisation can be found in Figure \ref{fig:sc_a}. Since $\KK$ is a discrete and combinatorial entity, one can compute meaningful topological information from  topological spaces (or datasets) described by simplicial complexes.

\subsubsection{Homology}
One topological invariant we can compute from simplicial complexes is homology. In each dimension $k$, the dimension of the $k$-th homology group can be thought of as the number of voids in a simplicial complex enclosed by a $k$-dimensional boundary. 
We restrict our definition of homology over the field of two elements, $\FF_2$, which is the setting for our computations. 
For a simplicial complex, the homology groups coincide with those of its geometric realisation (viewed as a topological space).

\begin{definition}
Let $\KK$ be a simplicial complex. We define its \emph{chain complex} $\mathcal{C}_\bullet(\KK)$ over $\FF_2$ to be the collection of vector spaces $\mathcal{C}_i=\FF_2\langle\KK_i\rangle$, together with the collection of linear maps $\partial_i:\mathcal{C}_i\to\mathcal{C}_{i-1}$ induced by
$$\partial_i: \tau\mapsto\sum_{v\in\tau}\tau\backslash\{v\}$$
for all $\tau\in\KK_i$.
\end{definition}

We observe that $\partial_i\circ\partial_{i+1}=0$ for all $i$. Furthermore, we note that any simplicial map $h:\KK\to\KK'$ induces a collection of maps on corresponding chain complexes $\mathcal{C}_\bullet$ and $\mathcal{C}'_\bullet$, denoted $\{\hat{h}_i:\mathcal{C}_i\to\mathcal{C}'_i\}_i$, which are defined as
$$\hat{h}_i(\tau):=\begin{cases}h(\tau) & \text{if } \dim h(\tau)=\dim\tau\\0 & \text{otherwise}\end{cases}.$$
We call such a collection of maps a \emph{chain map} from $\mathcal{C}_\bullet$ to $\mathcal{C}'_\bullet$. It satisfies $\partial'_i\circ \hat{h}_i=\hat{h}_{i-1}\circ\partial_i$ for all $i$.

\begin{definition}
Let $\KK$ be a simplicial complex and let $\mathcal{C}_\bullet(\KK)$ be its associated chain complex over $\FF_2$. Then the $k$-th homology group of $\KK$ is defined to be the quotient of vector spaces
$$H_k(\KK):=\frac{\ker\partial_i}{\mathrm{im}\,\partial_{i+1}}.$$
\end{definition}

Note that for $\hat{h}:\mathcal{C}_\bullet(\KK)\to \mathcal{C}_\bullet(\KK')$ the induced map $h^*:H_k(\KK)\to H_k(\KK')$ given by $h^*:[c]\mapsto [\hat{h}_k(c)]$, where $c\in\ker \partial_k$ and the brackets denote equivalence up to translation by $\mathrm{im}\,\partial_k$ and $\mathrm{im}\,\partial'_k$ respectively, is well defined for all $k$ \cite{roadmap}. Moreover, for simplicial maps $h:\KK\to\KK'$ and $h':\KK'\to\KK''$ we have $(h \circ h')^*= h^*\circ (h')^*$. This property is called the \emph{functorality of homology} and will be used when we introduce persistence.

\subsubsection{Persistence}

We view point clouds as a discrete subset of a continuous geometric object embedded in Euclidean space. The underlying continuous space is the primary subject of interest. In order to obtain information about this geometric object, we wish to inflate our discrete points to a continuous space, or to capture a relative offset between points in this space. In practice, we usually do not know the adequate inflation resolution. Persistence theory offers an elegant way to overcome this caveat by scaling the resolution from fine to coarse, and tracking how the homology of these spaces evolves by considering their canonical inclusion relations.

\begin{definition}
Let $\KK$ be a simplicial complex and let $g:\KK\to\RR$ be a function such that $\tau\subseteq\tau'$ implies $g(\tau)\leq g(\tau')$ for any $\tau,\tau'\in\mathcal{K}$. A \emph{filtration} of the simplicial complex $\KK$ by $g$ is then defined to be the sequence of simplicial complexes $\{\KK_L\}_{L\in\RR}$, where
$$\KK_L:=\{\tau\in\KK\,\vert\,g(\tau)\leq L\},$$
together with the canonical inclusions $\iota_L^{L'}:\KK_L\hookrightarrow\KK_{L'}$ whenever $L\leq L'$. An example of a filtration is visualised in Figure \ref{fig:sc_b}
\glsxtrnewsymbol[description={A map defining a filtration of a simplicial complex or topological space}]{g}{\ensuremath{g}}
\glsxtrnewsymbol[description={A topological space}]{Tcal}{\ensuremath{\mathcal{T}}}
In the same spirit, let $\mathcal{T}$ be a topological space and $g:\mathcal{T}\to\RR$ be a continuous function. A filtration of the topological space $\mathcal{T}$ is then defined to be the sequence of topological spaces $\{\mathcal{T}_L\}_{L\in\RR}$, where
$$\mathcal{T}_L:=\{x\in\mathcal{T}\,\vert\,g(x)\leq L\},$$
together with the canonical inclusions $\iota_L^{L'}:\mathcal{T}_L\hookrightarrow \mathcal{T}_{L'}$ whenever $L\leq L'$.
\end{definition}

A common way of constructing a filtration from a point cloud $\mathbf{v}\subset\RR^d$ is to set $\KK=\mathcal{P}(X)$ and $g(\tau)=\max\{d(x,y)\,\vert\,x,y\in \tau\}$. This is called the \emph{Vietoris-Rips filtration}, and $\KK_L$ is a good approximation to an inflation of $\mathbf{v}$ by placing balls of radius $L/2$ at each point \cite{oudot}. We will consider the following alternative filtration. For a fixed $L\in\RR$ and map $p:\RR^d\to\RR$, we set $\KK':=\KK_L$ in the Vietoris-Rips sense and consider the filtration by the map $g':\KK'\to\RR$ defined by $g'(\tau):=\max\{p(x)\,\vert\,x\in\tau\}$.

\begin{definition}
Let $\FF_2[t]$ be the ring of polynomials in the indeterminate $t$ with coefficients in $\FF_2$.
Let $\{\KK_L\}_{L\in\RR}$ be a filtration of a simplicial complex. Moreover, define $\mathrm{Crit}_L:=\{L\in\RR\vert\iota^L_{L-\varepsilon}\neq\mathrm{id}\,\forall \varepsilon>0\}$, the set of all $L$ at which $\KK_L$ changes (which is a finite set at $\KK$ is finite). Define the function $c:\mathbb{N}_0\to \mathrm{Crit}_L\cup\{\inf \mathrm{Crit}_L-1\}$ by mapping 0 to $\inf \mathrm{Crit}_L-1$ and $n>0$ to the $n$-th smallest element of $\mathrm{Crit}_L$ (without loss of generality, we map integers bigger than the cardinality of $\mathrm{Crit}_L$ to the largest element of $\mathrm{Crit}_L$).

For a fixed integer $k$, let $H_k(\,\cdot\,)$ denote the $k$-th simplicial homology with coefficients in $\FF_2$. Define
\begin{equation}
M_k:=\bigoplus_{n\in\NN_0}H_k\left(\KK_{c(n)}\right)\label{eq:pmod}\end{equation}
together with the action of $\FF_2[t]$ on $M_k$ induced by $t^a\cdot x = \iota_{c(n+a)}^{c(n)}(x)^*\in H_k(\KK_{c(n+a)})$ for $x\in H_k(\KK_{c(n)})$ and non-negative integer $a$. Then $M_k$ is a (graded) $\FF_2[t]$-module, called the \emph{persistence module} of the filtration.
\end{definition}
\glsxtrnewsymbol[description={The $k$-th homology functor}]{Hp}{\ensuremath{H_k}}
\glsxtrnewsymbol[description={A persistence module}]{Mc}{\ensuremath{M}}
\glsxtrnewsymbol[description={A matching of barcodes}]{m}{\ensuremath{m}}

The definition works analogously for a filtration of a topological space (assuming that the homology of the spaces changes at only finitely many filtration values). It can be shown that the operation of taking a persistence module of a filtration of a simplicial complex (or a topological space) is functorial. Hence, persistence modules are algebraic invariants of filtrations.

Since $\KK$ is finite, the persistence module $M_k$ is finitely generated as a $\FF_2[t]$-module. As $\FF_2[t]$ is a principal ideal domain, $M_k$  decomposes into summands generated by a single object uniquely up to (graded) isomorphism and permutation of summands. Hence, we can write
$$M_k\cong\left(\bigoplus_{a\in G_F}\FF_2[t]\right)\oplus\left(\bigoplus_{b\in G_T}\FF_2[t]/\langle t^{d_b}\rangle\right),$$
where $G_F$ is the subset of chosen generators that are free and $G_T$ is the subset of generators that are torsion. In particular, any element in $G_F$ or $G_T$ will have a non-zero entry in exactly one summand of the decomposition in Equation (\ref{eq:pmod}). We call the integer $n$ indexing this entry the \emph{degree} of that element.

\begin{definition}
Let $M_k$ be a persistence module that decomposes as above. Let $\mathrm{deg}:G_F\cup G_T\to\NN_0$ be the function mapping each element to its degree. The barcode of $M_k$ is defined to be the multiset
$$\mathcal B:=\{(c(\mathrm{deg}(a)),\infty)\,\vert\,a\in G_F\}\cup\{(c(\mathrm{deg}(a)),c(\mathrm{deg}(a)+d_a))\,\vert\,a\in G_T\}.$$
We call the elements of $\mathcal{B}$ \emph{bars}, the first coordinate of each bar its \emph{birth-value}, the latter coordinate its \emph{death-value} and the absolute difference of the coordinates its \emph{persistence}.

A \emph{matching of barcodes} $\mathcal B$ and $\mathcal B'$ is a partial injection $\varpi:\mathcal B\hookrightarrow\mathcal B'$.
The \emph{bottleneck distance} between $\mathcal B$ and $\mathcal B'$ is defined to be
$$d_{BD}\left(\mathcal B,\mathcal B'\right):=\inf_\varpi\,\max\left\{\max_{a\in\mathrm{dom}\,\varpi}\left\|a-\varpi(a)\right\|_\infty,\max_{(x,y)\not\in\mathrm{dom}\,\varpi}\frac{y-x}{ 2},\max_{(x,y)\not\in\mathrm{im}\,\varpi}\frac{y-x}{ 2}\right\},$$
where the infimum is taken over all possible matchings and elements of a barcode are viewed as elements of $\RR^2$ (we assume $\infty-\infty=0$). Here, $\mathrm{dom}\, \varpi$ is the domain of $\varpi$, i.e., the set of inputs at which $\varpi$ is defined.
\end{definition}

The bottleneck distance defines a metric on the space of barcodes \cite{oudot}. This metric is stable in the following sense:

\begin{theorem}[e.g. Corollary 3.6 in \cite{oudot}]
Let $\KK$ be a simplicial complex and let $g,g':\KK\to\RR$ be functions defining filtrations of $\KK$, and subsequently persistence modules $M_k$ and $M'_k$, and barcodes $\mathcal B$ and $\mathcal B'$. Then
$$d_{BD}\left(\mathcal B,\mathcal B'\right)\leq\left\|g-g'\right\|_\infty.$$
\end{theorem}

Henceforth, we write  $\mathrm{PH}_k(g)$ to denote the $k$-dimensional persistent homology (which can equivalently be summarised by a barcode or a persistence module) of a simplicial complex or a topological space filtered by a function $g$.

\subsubsection{Persistent homology of random data}
In this Subsection, we study the persistent homology of the posterior distributions of the parameter inferences of Subsection \ref{sec:inference}.
Note that simplicial complexes, filtrations and persistent homology can also be employed to compare biological models \emph{a priori} (i.e., with no dependence on measurement data) \cite{vittadello2020model}.

We demonstrate that filtering a Vietoris-Rips complex for a fixed value $L$ by a function $g'$, as described at the beginning of this Section, yields more discriminative power. Here, we pick $g'$ to be an estimated probability density function. These filtrations turn out to be highly discriminative between the mutants and offer novel insight at the biological level.
While a Vietoris-Rips filtration is entirely based on distances, the construction we employ, using a Vietoris-Rips complex at a fixed parameter and then filtering it by a probability density function (pdf), places an emphasis on density. The information encoded is directly related to the probability distribution and the resulting barcodes will stabilize as the sample size increases \cite[Theorem 3.5.1 in ][]{Rabadan2020}. Furthermore, the chosen construction is stable with respect to outliers. By contrast, in a Vietoris-Rips filtration, bars in the resulting barcodes will converge towards zero length when increasing the sample size and a single outlier, even in a large sample, can change a barcode drastically.

Initially, assume that we are given a probability density function $p:\RR^m\to\RR$. This defines a filtration of the graph by $-p$, $\mathcal{T}$ say, via $\mathcal{T}_L=\left\{x\in\RR^m\,\vert\, -p(x)\leq L\right\}$. For $L'\leq L$ we then have $\mathcal{T}_{L'}\subseteq \mathcal{T}_{L}$. Such a filtration is visualised for the case $m=1$ in Figure \ref{fig:pdffilt}. By analogy with filtrations of simplicial complexes, we can theoretically compute a barcode for each such topological filtration and investigate the resulting bottleneck distances.

For each (homological) dimension, these barcodes provide a topological signature of a posterior distribution. We point out that although this signature is \emph{not} a sufficient statistic, it is effective at distinguishing between posteriors corresponding to distinct mutants in our application. In particular, for any pdf $p_1:\RR^d\to\RR$, the pdf $p_2(x)=p_1(x-x_0)$ gives rise to the same topological signature for any constant $x_0\in\RR^d$. Thus, rather than comparing the location of probability density in parameter space, in the context of a Bayesian inference, this topological signature captures the quality of the certainty we have in parameter values, irrespective of their location.

For example, bars in the $H_0$-barcode encode the density (as negative of the birth-value) and the prominence (as the persistence) of the modes of a pdf. Similarly, Morse Theory tells us that for a (smooth) pdf on $\RR^d$, the $(d-1)$-th barcode captures local minima by their density (as death-value) and the depth of their basin of attraction (as persistence).

In order to conduct such a topological analysis, two questions must be addressed:
\begin{enumerate}
\item How can we approximate the topology of a graph of a probability density combinatorially (i.e., in a manner amenable to the application of discrete computational methods) if only point samples are available?
\item Can we test the statistical significance of the resulting bottleneck distances?
\end{enumerate}

To resolve the first question, we will employ a result from Bobrowksi et al. \cite{kerest} that relies on the concept of kernel density estimation (KDE). In order to test the significance of the resulting bottleneck distance, we will use an empirical p-value estimate.

\begin{definition}
Let $\mathbf{v}=\{v_1,\dots,v_N\}\subseteq\RR^m$ be a set of $N$  samples drawn independently from a probability distribution governed by the density function $p:\RR^m\to\RR$. Let $K:\RR^m\to\RR$ be smooth, unimodal, symmetric probability density function whose support is contained in the unit ball centred at $0$. Then
$$\hat{p}_b(x)=\frac{1}{ Nb^m}\sum_{i=1}^NK\left(\frac{x-v_i }{ b}\right)$$
is called a \emph{kernel density estimate (KDE)} of $p$ with \emph{bandwidth} $b$.
\end{definition}
\glsxtrnewsymbol[description={A kernel function}]{Kc}{\ensuremath{K}}%
\glsxtrnewsymbol[description={The bandwidth of a kernel}]{b}{\ensuremath{b}}%

On each sample $v_i$, we place a pdf and average it, where $b$ controls the width of each pdf, that is, how much of the probability mass is centred around $v_i$. Loosely speaking, if $b$ is too large, then the resulting function underfits a histogram given by the data, while if it is too small, then the bandwidth overfits the histograms (see Figure \ref{fig:kde_eg}). The bandwidth is negatively correlated with the sample size and there are standardised ways of picking optimal bandwidths for the case where $p$ is unknown \cite{Henderson2012}.

\begin{figure}
\centering
\includegraphics[width=0.6\textwidth]{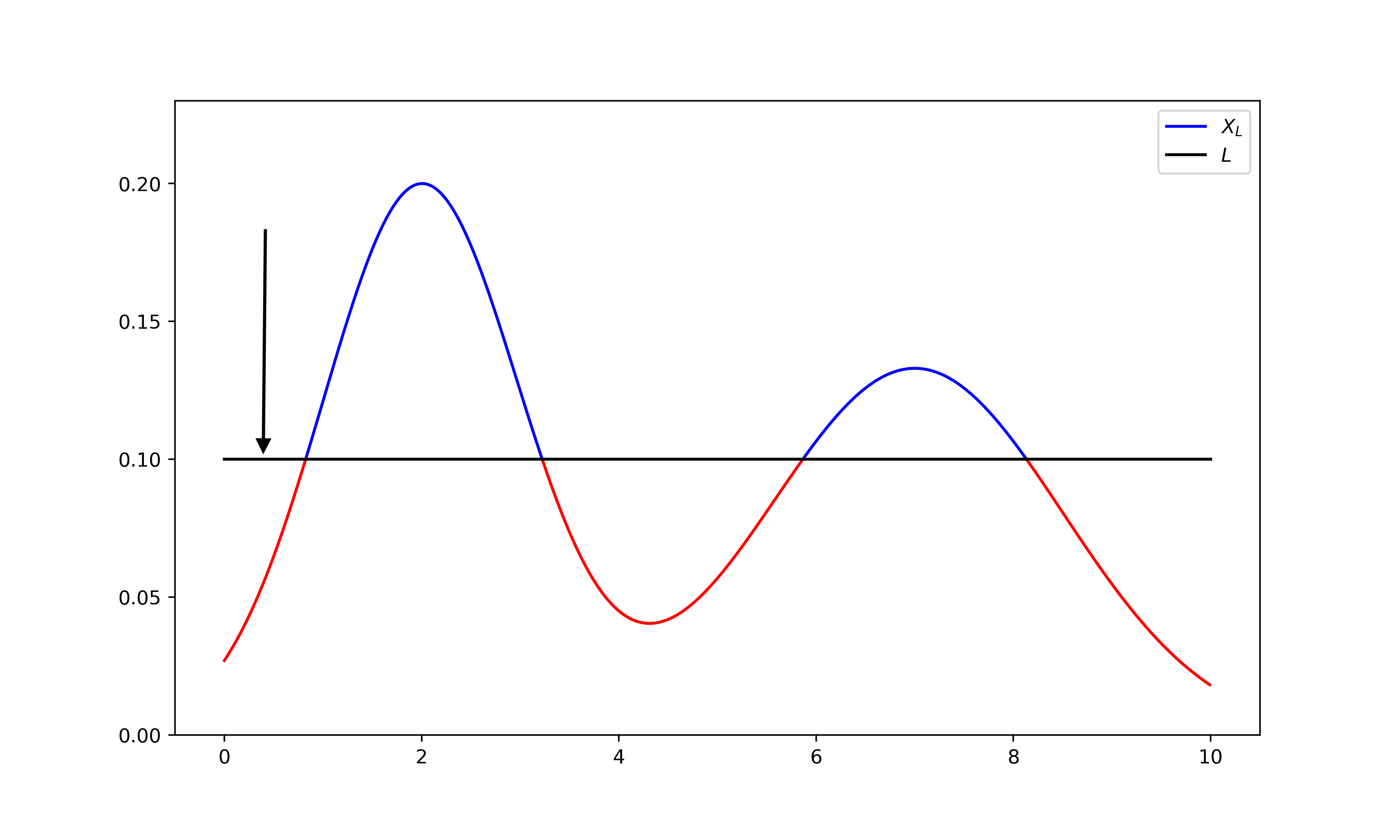}
\caption{An example of a super-level-set filtration of the graph of a density function $p:\RR^1\to\RR$. This is equivalent to a sub-level-set filtration of $-p$.}\label{fig:pdffilt}
\end{figure}

\begin{figure}
\centering
\includegraphics[width=10cm]{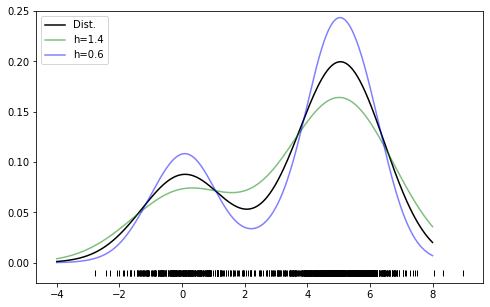}
\caption{A probability density function on $\RR^1$ (black line) with 1.000 samples (black dashes). We see kernel density estimations with bandwidth 0.6 (blue line) and 1.4 (green line). The ideal bandwidth is approximately 1.}\label{fig:kde_eg}
\end{figure}

Given such an i.i.d. sample $\mathbf{v}=\{v_1,\dots,v_N\}\subseteq\RR^m$ from our probability density function $p$ and an optimal bandwidth $b$, we can construct a Vietoris-Rips complex with fixed parameter $b$ (equalling the bandwidth)
$$\mathrm{VR}_b(\mathbf{v}):=\left\{\{v_0,\dots,v_k\}\subseteq \mathbf{v}\,\vert\,\|v_i-v_j\|\leq b\,\forall i,j\right\}.$$
For the sake of brevity, let \glsxtrnewsymbol[description={A Vietoris-Rips complex on vertices $\mathbf{v}$ at resolution $b$}]{phat}{\ensuremath{\mathrm{VR}_b(\mathbf{v})}}%
$\KK=\text{VR}_b(\mathbf{v})$. The KDE $\hat{p}_b$ of $p$ based on $\mathbf{v}$ then extends to a function on $\mathcal K$ via
$$\hat{p}_b(\{v_0,\dots,v_k\}):=\min\left\{\hat{p}_b(v_0),\dots,\hat{p}_b(v_k)\right\}.$$
In return, the extended function $\hat{p}_r$ defines a filtration $\{\KK_L\}_{L\in\RR}$ of $\KK$ by
$$\KK_L:=\left\{\{v_0,\dots,v_k\}\,\vert\,-\hat{p}_b(\{v_0,\dots,v_k\})\leq L\right\}.$$
We seek to relate the persistent homology of the filtration of simplicial complexes $\KK_L$ to the persistent homology of the filtration of topological spaces $\mathcal{T}_L$.

\glsxtrnewsymbol[description={A barcode}]{B}{\ensuremath{\mathcal{B}
}}

In order to use results from \cite{kerest}, we introduce some notation. For a function $f:\KK\to\RR$ and $\eta>0$ define
$f_{\lfloor\eta\rfloor}(\sigma):=2\eta\lfloor f(\sigma)/(2\eta)\rfloor$. Then

\begin{theorem}[Theorem 3.7 in \cite{kerest}]\label{thm:kerest}
Let $p:\RR^m\to\RR$ be a smooth bounded pdf with finitely many critical points. Let $\hat{p}$ be a KDE with bandwidth $b$ based on $n$ i.i.d samples of $p$ and $\KK$ be a simplicial complex as above. Assume $b\to0$ and $Nb^m\to\infty$. Then for any $0\leq k \leq m$, we have
$$\mathrm{Pr}\left(d_{BD}\left(\mathrm{PH}_k(p\right), \mathrm{PH}_k\left(\hat{p}_{\lfloor\eta\rfloor}\right)\leq 5\eta \right)\geq 1-3\eta^* Ne^{-C_\eta Nr^d},$$
where for $p_\mathrm{max}:=\sup_{x\in\RR^d}p(x)$ we define
$$\eta^*:=\left\lceil p_\mathrm{max}/2\eta\right\rceil\quad\text{and}\quad C_\eta:=\frac{(\eta/2)^2}{ 3p_\mathrm{max}+\eta/2}.$$
\end{theorem}

Theoretically, the above theorem can be exploited for testing the null hypothesis $\mathbf{H}_0:\mathrm{PH}_k\left(p\right)=\mathrm{PH}_k\left(p'\right)$ for two distributions $P$ and $P'$ with associated densities $p$ and $p'$, as the result enables us to establish a bound on how large  
a bottleneck distance can be explained by sampling noise at a given significance level. However, we estimate that to use this theorem for showing that the bottleneck distances between posterior distributions associated with the wild-type and the four mutants are significant, we must sample at least $1.5\times 10^7$ points per distribution. This makes persistent homology computation infeasible.

At the same time, we observe that there is little change in the bottleneck distances between the barcodes resulting from the wild-type's and the four mutants' posterior distributions when resampling point clouds containing as few as 200 points. This leads us to think that the true p-value associated with the null hypothesis $\mathbf{H}_0:\mathrm{PH}_k\left(p\right)=\mathrm{PH}_k\left(p'\right)$, where $p$ and $p'$ are posterior densities corresponding to the wild-type and a mutant is possibly much lower than the upper bound derived by appealing to Theorem \ref{thm:kerest}. 
One factor that may explain this discrepancy is that while our distributions are technically distributions on $\RR^3$, they have compact support. Similarly, major sources of instability for KDE, and subsequently for the filtration of density functions, are modes linked to outliers, while repeated simulations suggest that in our case all density functions are
unimodal. Together, these aspects imply that the computed barcodes could converge to the barcode obtained by filtering the unknown density function at a faster rate than in the general setting of Theorem \ref{thm:kerest}.

Henceforth, we use the method of constructing a filtration based on a point cloud proposed in \cite{kerest}, which is provably well-behaved asymptotically but uses a different approach to estimate significance. To do this we opt for a Monte Carlo p-value estimate, also known as the empirical p-value (e.g. see \cite{bsapp}). For each mutant, we sample $\beta$ additional point clouds of size $n$ from the posterior distribution. In this context, for the first mutant (or the wild-type) under investigation, call the original point cloud $\mathbf{v}$ and let $\mathbf{v}_i$ for $i=1,...,\beta$ denote $\beta$ additional point clouds of size $n$, obtained by repeated sampling. Define $\mathbf{v}'$ and $\mathbf{v}'_i$ analogously for a distinct mutant. Let $d_i=d_{BD}\left(\mathrm{PH}_k\left(\hat{p}\right),\mathrm{PH}_k\left(\hat{p}_i\right)\right)$, where $\hat{p}_i$ is the density estimate obtained from $\mathbf{v}_i$ and define $d_i'$ analogously. Assume $d=d_{BD}\left(\mathrm{PH}_k\left(\hat{p}\right),\mathrm{PH}_k\left(\hat{p}'\right)\right)$ is the $j$-th largest element in the multiset $\left\{d_i\right\}_{i=1}^\beta\cup\{d\}$ and the $j'$-th largest element in $\left\{d'_i\right\}_{i=1}^\beta\cup\{d\}$ for two distinct mutants, then
$$\hat{\pi}=\min\left\{\frac{\beta+1-j}{ \beta+1},\frac{\beta+1-j' }{ \beta+1}\right\}$$%
\glsxtrnewsymbol[description={A p-value estimate}]{varpi}{\ensuremath{\hat{\pi}}}%
is a p-value estimate for a hypothesis test $\mathbf{H}_0:\mathrm{PH}_1\left(p\right)=\mathrm{PH}_1\left(p'\right)$. The resulting p-value estimates, for each pair of mutants and wild-type, can be found in Table \ref{fig:pval}. It is likely that these p-value estimates over-estimate the actual value, but they allow us to reject all null hypotheses at a significance level of $0.05$ \cite{North2002}.

\begin{table}
\centering
\begin{tabular}{c||c|c|c|c|c}
$d_{BD}$ & wild-type & Y130C & F53S & E203K & SSDD\\\hline\hline
wild-type & 0.0000 & 401.5999 & 334.7258 & 186.3972 & 2162.7175 \\\hline
Y130C & 401.5999 & 0.0000 & 401.5999 & 401.5999 & 2124.4453\\\hline
F53S & 334.7258 & 401.5999 & 0.0000 & 334.7258 & 2162.7175\\\hline
E203K & 186.3972 & 401.5999 & 334.7258 & 0.0000 & 2162.7175 \\\hline
SSDD & 2162.7175 & 2124.4453 & 2162.7175 & 2162.7175 & 0.0000 \\\hline
\end{tabular}\\\vspace{\baselineskip}
\begin{tabular}{c||c|c|c|c|c}
$\hat{\pi}$ & wild-type & Y130C & F53S & E203K & SSDD\\\hline\hline
wild-type & 0 & 0.01 & 0.01 & 0.01 & 0.01 \\\hline
Y130C & 0.01 & 0 & 0.01 & 0.01 & 0.01 \\\hline
F53S & 0.01 & 0.01 & 0 & 0.01 & 0.01 \\\hline
E203K & 0.01 & 0.01 & 0.01 & 0 & 0.01 \\\hline
SSDD & 0.01 & 0.01 & 0.01 & 0.01 & 0 \\\hline
\end{tabular}
\caption{Bottleneck distance between Barcodes for $H_1$ obtained from super-level-set filtration of KDEs (top) and their respective p-value estimates (bottom).}\label{fig:pval}
\end{table}

The results (Table \ref{fig:pval}) of the topological data analysis quantifies the differences between the Linear ERK model parameter posteriors for WT and mutants, and find SSDD mutant kinetics are most different from WT and other mutants.
This biological result raises the suitability for using the SSDD variant as a replacement for wild-type MEK activated by Raf. We suggest this should be investigated with further experimental studies. The previous work by Yeung \etal \cite{Yeung2020} found that $\pi$, the processivity parameter, of E203K was differed the most from WT MEK. Here we extended and complemented their analysis by comparing the three parameters together as a point cloud.

It remains to address the practical computability of all the constructions involved. As mentioned in the previous Section, we use the statistical software STAN (in particular, PyStan) to sample from the posterior distributions. This sampling is approximate via Hamiltonian MCMC \cite{STAN}, but we can verify via output summaries and trace plots of the Markov chains involved that all chains have converged close to their stationary distribution during the warm-up phase. 

In order to construct the KDE, we used the \texttt{KernelDensity} method of the Python package \texttt{sklearn}. We use the Epanechnikov kernel, which satisfies the hypothesis of the kernel in Theorem \ref{thm:kerest}. As a guess for the bandwidth, this package uses a rule of thumb proportional to Silverman's method, which we then cross-validate and plot against a histogram of our samples for each marginal distribution. Given experimental data, we construct a Vietoris-Rips complex with a radius $b$, equalling the bandwidth from the KDE, using the Python package \texttt{dionysus} (Version 2). We compute the resulting bottleneck distances using the package \texttt{persim}.

\section{Conclusion}

We presented an exhaustive mathematical analysis that supports the three main findings presented in Yeung et al \cite{Yeung2020}: model reduction, analysis of the model parameters and comparing mutation kinetics. Yeung et al. observed that certain values of parameter combinations from the Full ERK model fit the data, which in turn motivated the creation of a reduced model, the Linear ERK model. We confirmed the derivation of the Linear ERK model using algebraic QSS and the validity of the QSS approximation using the QSS variety. 
We performed systematic identifiability analyses on all three models, which is a prerequisite for meaningful parameter estimation. We found the Full, Rational and Linear ERK models are structurally identifiable. We then improved a previous definition of practical identifiability and showed that the Linear ERK model is practically identifiability but Rational and Full ERK models are not, which is consistent with \cite{Yeung2020}. We reproduced the parameter inference for wild-type and mutant MEK experiments.  While Yeung et al visually inspected samples of the posteriors, here we quantified these point clouds with computational algebraic topology. In future, 
it would be interesting to further explore the relationship between topological analysis and practical identifiability and how they may be used to inform experimental design \cite{apgar,hagen}. 
Throughout we stated open problems, which showcase the potential role of algebra, geometry and topology in systems and synthetic biology. Complementary to the analysis here is an inference of models in systems and single-cell biology that relies on algebra and topology \cite{wang2019inferring, StumpfMC2021, rizvi2017single}.
We hope 
our analysis of this ERK case study will motivate other systems biologists to partner with algebraists and topologists to analyse dynamical systems together with their experimental setup and data.

\section*{Acknowledgements}
The authors thank Gleb Pogudin and Stas Shvartsman for helpful discussions.
HAH, LM and HMB are grateful to the support provided by the UK Centre for Topological Data Analysis EPSRC grant EP/R018472/1. HAH gratefully acknowledges funding from EPSRC EP/R005125/1 and EP/T001968/1, the Royal Society RGF$\backslash$EA$\backslash$201074 and UF150238, and Emerson Collective. LM gratefully acknowledges support by the EPSRC grant EP/R513295/1 and by the Ludwig Institute for Cancer Research.

\renewcommand\refname{References}
\cleardoublepage\phantomsection\addcontentsline{toc}{chapter}{Bibliography}
\printbibliography

\renewcommand\refname{APPENDIX}

\renewcommand\thesection{\Alph{section}}
\setcounter{section}{1}
\setcounter{subsection}{0}
\renewcommand*{\thesubsection}{\Alph{section}.\the\value{subsection}}

\pagebreak

\section*{Appendix}\addcontentsline{toc}{section}{Appendix}

\subsection{QSSA Model Reduction via the Classical Singular Perturbation Theory Approach}\label{sec:qssa_mr}

Recall the Full ERK model (Equations (\ref{eq:a_start}) to (\ref{eq:a_end})):

\begin{flalign}\notag
\frac{\mathrm dS_0}{\mathrm dt}&= -k_{f_1}E\cdot S_0+k_{r_1}C_1,\\\notag
\frac{\mathrm dC_1}{\mathrm dt}&= k_{f_1}E\cdot S_0 -(k_{r_1}+k_{c_1})C_1,\\\notag
\frac{\mathrm dC_2}{\mathrm dt}&= k_{c_1}C_1-(k_{r_2}+k_{c_2})C_2 + k_{f_2}E\cdot S_1,\\\notag
\frac{\mathrm dS_1}{\mathrm dt}&= -k_{f_2}E\cdot S_1+k_{r_2}C_2,\\\notag
\frac{\mathrm dS_2}{\mathrm dt}&=k_{c_2}C_2,\\\notag
\frac{\mathrm dE}{\mathrm dt}&= -k_{f_1}E\cdot S_0+k_{r_1}C_1-k_{f_2}E\cdot S_1 +(k_{r_2}+k_{c_2})C_2,
\end{flalign}
with initial conditions $S_0=S_{tot}$, $E=E_{tot}$, and $S_1=S_2=C_1=C_2=0$ at time $t=0$. Note that the ODE for $S_2$ decouples from the system.
It is straightforward to show that the system satisfies the following conservation laws:
\begin{flalign}\notag
S_{tot}&=S_0+S_1+S_2+C_1+C_2,\\\notag
E_{tot}&=E+C_1+C_2.
\end{flalign}

We then non-dimensionalise by substituting $\tilde{S}_i:=S_i/S_{tot}$ for $i=0,1,2$; $\tilde{C}_j:=C_j/E_{tot}$ for $j=1,2$; and $\tilde{E}=1-\tilde{C}_1-\tilde{C}_2$ which is analogous to Equation (\ref{eq:e1}) and removes one equation from the ODE system. We also observe that $S_2$ decouples, leaving us with a system of four differential equations. Then
\begin{flalign}\notag
\frac{\mathrm d\tilde{S}_0}{\mathrm dt}&= -k_{f_1}E_{tot}\tilde{E}\cdot \tilde{S}_0+k_{r_1}\frac{E_{tot}}{ S_{tot}}\tilde{C}_1,\\\notag
\frac{\mathrm d\tilde{C}_1}{\mathrm dt}&= k_{f_1}S_{tot}\tilde{E}\cdot \tilde{S}_0 -(k_{r_1}+k_{c_1})\tilde{C}_1,\\\notag
\frac{\mathrm d\tilde{C}_2}{\mathrm dt}&= k_{c_1}\tilde{C}_1-(k_{r_2}+k_{c_2})\tilde{C}_2 + k_{f_2}S_{tot}\tilde{E}\cdot \tilde{S}_1,\\\notag
\frac{\mathrm d\tilde{S}_1}{\mathrm dt}&= -k_{f_2}E_{tot}\tilde{E}\cdot \tilde{S}_1+k_{r_2}\frac{E_{tot}}{ S_{tot}}\tilde{C}_2.
\end{flalign}

Subsequently, substituting $T:=(E_{tot}k_{f_1})\cdot t$ for $t$:

\begin{flalign}\notag
\frac{\mathrm d\tilde{S}_0}{\mathrm dT}&= -\tilde{E}\cdot \tilde{S}_0+k_{r_1}\frac{1}{ S_{tot}k_{f_1}}\tilde{C}_1,\\\notag
\frac{\mathrm d\tilde{C}_1}{\mathrm dT}&= \frac{S_{tot}}{ E_{tot}}\tilde{E}\cdot \tilde{S}_0 -\frac{1}{ E_{tot}k_{f_1}}(k_{r_1}+k_{c_1})\tilde{C}_1,\\\notag
\frac{\mathrm d\tilde{C}_2}{\mathrm dT}&= \frac{k_{c_1}}{ E_{tot}k_{f_1}}\tilde{C}_1-\frac{k_{r_2}+k_{c_2}}{ E_{tot}k_{f_1}}\tilde{C}_2 + \frac{k_{f_2}}{ k_{f_1}}\frac{S_{tot}}{E_{tot}}\tilde{E}\cdot \tilde{S}_1,\\\notag
\frac{\mathrm d\tilde{S}_1}{\mathrm dT}&= -\frac{k_{f_2}}{ k_{f_1}}\frac{S_{tot}}{ E_{tot}}\tilde{E}\cdot \tilde{S}_1+\frac{k_{r_2}}{S_{tot}k_{f_1}}\tilde{C}_2.
\end{flalign}

Then
\begin{flalign}\notag
\varepsilon\frac{\mathrm d\tilde{C}_1}{\mathrm dT}&= \tilde{E}\cdot \tilde{S}_0 -\frac{1}{ S_{tot}k_{f_1}}(k_{r_1}+k_{c_1})\tilde{C}_1,\\\notag
\varepsilon\frac{\mathrm d\tilde{C}_2}{\mathrm dT}&= \frac{k_{c_1}}{ S_{tot}k_{f_1}}\tilde{C}_1-\frac{k_{r_2}+k_{c_2}}{ S_{tot}k_{f_1}}\tilde{C}_2 + \frac{k_{f_2}}{ k_{f_1}}\tilde{E}\cdot \tilde{S}_1,
\end{flalign}
where $\varepsilon:=E_{tot}/S_{tot}$. As $E_{tot}\ll S_{tot}$, we have $\varepsilon\approx0$ and assume $\varepsilon\to0$. This gives

\begin{flalign}\notag
\tilde{C}_1&=\frac{k_{f_1}S_{tot}K_2\tilde{S}_0}{ ((k_{f_1}\tilde{S}_0+k_{f_2}\tilde{S}_1)k_{c_1}+\tilde{S}_0k_{f_1}K_2+\tilde{S}_1k_{r_1}k_{f_2})S_{tot}+K_1K_2},\\\notag
\tilde{C}_2&=\frac{((k_{f_1}\tilde{S}_0+k_{f_2}\tilde{S}_1)k_{c_1}+k_{r_1}k_{f_2}\tilde{S}_1)S_{tot}}{ (S_{tot}(k_{f_1}\tilde{S}_0+k_{f_2}\tilde{S}_1)+k_{c_2}+k_{r_2})k_{c_1}+(k_{r_1}k_{f_2}\tilde{S}_1+k_{f_1}K_2\tilde{S}_0)S_{tot}+k_{r_1}K_2},
\end{flalign}
where $K_i:=k_{r_i}+k_{c_i}$.

Finally, by substituting the above expressions for $\tilde{C}_1$ and $\tilde{C}_2$:
\begin{flalign}\notag
\frac{\mathrm d\tilde{S}_0}{\mathrm dT}&= \frac{-k_{c_1}K_2\tilde{S}_0}{  S_{tot}(k_{c_1}+k_{c_2)}(k_{f_1}\tilde{S}_0+k_{f_2}\tilde{S}_1)+K_1K_2},\\\notag
\frac{\mathrm d\tilde{S}_1}{\mathrm dT}&=\frac{-(k_{f_2}/k_{f_1})k_{c_2}K_1\tilde{S}_1+k_{r_2}k_{c_1}\tilde{S}_0}{ S_{tot}(k_{c_1}+k_{c_2)}(k_{f_1}\tilde{S}_0+k_{f_2}\tilde{S}_1)+K_1K_2} .
\end{flalign}

Then, by substituting $S_i$ for $\tilde{S}_i$, and dividing both numerator and denominator by $K_1K_2$
\begin{flalign}\notag
\frac{\mathrm dS_0}{\mathrm dt}&= \frac{-\kappa_1S_0}{ \gamma_1 S_0+\gamma_2S_1+1},\\\notag
\frac{\mathrm dS_1}{\mathrm dt}&=\frac{-\kappa_2S_1+(1-\pi)\kappa_1S_0}{ \gamma_1S_0+\gamma_2S_1+1},
\end{flalign}
where $\kappa_i$, $\pi$, and $\gamma_i:=(k_{c_1}+k_{c_2})k_{f_i}/(K_1K_2)$ are as defined previously in this paper.

This system reduces to the linear case if the denominator is approximately constant.
In particular, this is the case if $0<\gamma_1S_0+\gamma_2S_1\ll1$. This is the case, for our setting, which can be seen by writing 
$$\gamma_1=\frac{k_{c_1}+k_{c_2}}{ k_{c_2}+k_{r_1}}\frac{1}{ k_{M_1}}$$
(similarly for $\theta_2$), as we expect the first factor to be approximately 1 \cite{BarEven11} or smaller and $k_{M_1}$ (similarly $k_{M_2}$), the Michaelis-Menten constant, to be approximately $25\mu M$.

\subsection{QSSA Model Reduction via the Algebraic Approach}\label{sec:algebraic_mr}
We use the algebraic approach to derive two reduced models.

\subsubsection{Deriving the Rational ERK Model}
We substitute $E$ using Equation (\ref{eq:e1}) into the Full ERK model (Equations (\ref{eq:a})). Using the algebraic reduction, we follow the steps of Appendix \ref{sec:qssa_mr}, except instead of taking a limit $\varepsilon\to0$, simply set  $\varepsilon(\mathrm d\tilde{C}_1/\mathrm dT)=\varepsilon(\mathrm d\tilde{C}_2/\mathrm dT)= 0$, which leads to the Rational ERK model.
In other words, we treat the assumption that  $\varepsilon(\mathrm d\tilde{C}_1/\mathrm dT)$ and $\varepsilon(\mathrm d\tilde{C}_2/\mathrm dT)$ converge to zero as an algebraic fact and without regard of any analytic implications the assumption $\varepsilon\to0$ could have on the right-hand-side of these differential equations. Note that we do not need to non-dimensionalise as we did in Appendix \ref{sec:qssa_mr}.
We justify the accuracy of such an approximation Appendix \ref{sec:accuracy}. 

\subsubsection{Deriving the Linear ERK Model}
We now substitute $E$ using Equation (\ref{eq:e2}),
$$ E=E_{tot}-S_{tot}+S_0+S_1+S_2=: E',$$
into the Full ERK model (Equations (\ref{eq:a})).
 Setting  $\mathrm dC_1/\mathrm dt=0$ and $\mathrm dC_2/\mathrm dt=0$, 
we get
\begin{flalign*}
C_1&=\frac{k_{f_1}}{ k_{r_1}+k_{c_1}}E'\cdot S_0 , \\
C_2&=\frac{1}{ k_{r_2}+k_{c_2}}\left(k_{f_2}E'\cdot S_1+k_{c_1}C_1\right)=\frac{E'}{ k_{r_2}+k_{c_2}}\left(k_{f_2}S_1+\frac{k_{f_1}k_{c_1}}{ k_{r_1}+k_{c_1}}S_0\right),
\end{flalign*}
which can be substituted into $\mathrm dS_0/\mathrm dt$, $\mathrm dS_1/\mathrm dt$ and $\mathrm dS_2/\mathrm dt$:
\begin{flalign*}
\frac{\mathrm dS_0}{\mathrm dt}&=-\frac{k_{f_1}k_{c_1}}{ k_{r_1}+k_{c_1}}E'\cdot S_0,\\
\frac{\mathrm dS_1}{\mathrm dt}&=-\frac{k_{f_2}k_{c_2}}{ k_{r_2}+k_{c_2}}E'\cdot S_1+\frac{k_{r_2}}{ k_{r_2}+k_{c_2}}\frac{k_{f_1}k_{c_1}}{ k_{r_1}+k_{c_1}}E'\cdot S_0,\\
\frac{\mathrm dS_2}{\mathrm dt}&=\frac{k_{f_2}k_{c_2}}{ k_{r_2}+k_{c_2}}E'\cdot S_1+\frac{k_{c_2}}{ k_{r_2}+k_{c_2}}\frac{k_{f_1}k_{c_1}}{ k_{r_1}+k_{c_1}}E'\cdot S_0.
\end{flalign*}
We rewrite the above equations in terms of the parameter combinations $\kappa_1$, $\kappa_2$, and $\pi$ (see Equation (\ref{eq:pi1pi2kappa})) then gives  the Linear ERK model (Equations (\ref{eq:l1})--(\ref{eq:l3})).

Observe that the Linear ERK model obeys the conservation law $S_0+S_1+S_2=\tilde{S}_{tot}$. Here, $\tilde{S}_{tot}$ is a constant close to $S_{tot}$. Hence, $E'=E_{tot}-(S_{tot}-\tilde{S}_{tot})$ is a constant close to $E_{tot}$. The implication that $E'=E_{tot}$ at all times $t$ is physically infeasible. Whether $\tilde{S}_{tot}=S_{tot}$ depends on whether one updates the initial conditions for $S_j$ through an inner solution, which is typically assumed in singular-perturbation-theory approaches. The algebraic approach does not require computing an inner solution.

It is important to reiterate at this point that this reduction is not the result of a singular perturbation analysis. Indeed, if we had followed a singular perturbation analysis, similar to the one in Appendix \ref{sec:qssa_mr}, but with substitution (\ref{eq:e2}), we could not factor out $E_{tot}/S_{tot}$ in the non-dimensionalised differential equations for $C_1$ and $C_2$ as nicely as we were able to in the previous subsection. In this instance, we would have to leave factors of $\varepsilon^{-1}$ in the algebraic equations, which would be ambiguous when taking a limit $\varepsilon\to0$.

\subsection{Accuracy of Algebraic QSSA}\label{sec:accuracy}

We start by providing the full statement of Proposition 2 of \cite{reduction} (i.e., the full version of Proposition \ref{prop:accuracy} of this manuscript):

Let $K^*\subset\RR_+^n\times\RR_+^m$ satisfy the following:
\begin{itemize}
\item There exists $(\hat{y},\hat{\theta})$ in the interior of $K^*$ such that $f^{[2]}(\hat{y},\hat{\theta}) = 0$.
\item $D_2f^{[2]}(x, \theta)$ is invertible for all $(x,\theta)\in K^*$.
\item There exist $y_0\in\RR^n$ and $r > 0$ with the following property:\\
Whenever $(x, \theta)\in K^*$ for some $x\in\RR^n$ and some $\theta\in\RR^m_+$ then $\overline{B_r(y_0)}\times\{\theta\}\subseteq K^*$.
\item There exists an $R > 0$ such that $\|f(x,\theta)\|\leq R$ and $\|f_\mathrm{red}(x,\theta)\|\leq R$ for all $(x,\theta)\in K^*$.
\item There exists an $L > 0$ such that $\|Df(x,\theta)\|\leq L$ and $\|Df_\mathrm{red}(x,\theta)\|\leq L$ for all
$(x,\theta)\in K^*$.
\end{itemize}
These conditions imply that every $V_{\theta^*}$, with ${\theta^*}$ near $\hat{\theta}$, is a submanifold. Note
that every $(y, \theta^*)$ with $y$ in the interior of $\RR^n_+$ is contained in some $K^*$ that satisfies the last three of the above conditions.

\begin{proposition}\label{prop:full_prop}
Assume that the above conditions are satisfied for $K^*$.

(a) Let $\theta^*$ be given such that $V_{\theta^*}\times\{\theta^*\}$ has non-empty intersection with $\mathrm{int}\, K^*$, let $(y, \theta^*)$ be a point in this intersection and $V'_{\theta^*}\subset \RR^n$ be some open neighborhood of $y$ such that $(V_{\theta^*}\cap V'_{\theta^*})\times\{\theta^*\}\subset K^*$.
Moreover let $t^* > 0$
such that the solution of (\ref{eq:ode}) with initial value $y$ exists and remains in $V'_{\theta^*}$ for all $t\in[0, t^*]$. Then there exists a compact neighbourhood $ A_{\theta^*}\subseteq V'_{\theta^*}$ of $y$ with the following properties: (i) For every $z\in A_{\theta^*}$ the solution of
(\ref{eq:ode}) with initial value $z$ exists and remains in $V'_{\theta^*}$ for all $t\in[0, t^*]$. (ii)
For every $\varepsilon' > 0$ there is a $\delta_1 > 0$ such that the solution of (\ref{eq:expred}) with initial
value $z\in A_{\theta^*}\cap V_{\theta^*}$ exists and remains in $V_{\theta^*}$ for $t\in[0,t^*]$ whenever
$\|f-f_\mathrm{red}\| < \delta_1$ on $V'_{\theta^*}$. (iii) For every $\varepsilon' > 0$ there is a $\delta\in (0, \delta_1]$ such that the difference of the solutions of (\ref{eq:ode}) resp. of (\ref{eq:expred}) with initial
value $z\in A_{\theta^*}\cap V_{\theta^*}$ has norm less than $\varepsilon'$ for all $t\in[0, t^*]$ whenever
$\|f-f_\mathrm{red}\| < \delta$ on $V'_{\theta^*}$.

(b) Let $y\in V_{\theta^*}$ and let $\rho_0 > 0$ such that
$$\overline{B_{\rho_0/2L}(y)} \times \{{\theta^*}\}\subseteq K^*.$$
Let $\rho\leq\rho_0$ such that $\|f(y,{\theta^*})-f_\mathrm{red}(y,{\theta^*})\| \geq 2\rho$. Then for $t':= \rho/(2LR)$
the solutions of (\ref{eq:ode}) resp. of (\ref{eq:expred}) with initial value $y$ exist and remain
in $\overline{B_{\rho_0/2L}(y)}$ for $0\leq t\leq t'$, and their difference has norm at least $\rho^2/(2LR)$ at $t = t'$.
\end{proposition}

Returning to our two model reductions of interest, we use the notation Sections \ref{sec:model} and \ref{sec:algred} and label the two possible substitutions for the variable $E$ as (\ref{eq:e1}) and (\ref{eq:e2}), as before.
Moreover, we write $\dot{S_0}$ and $\dot{S_1}$ as shorthand for the algebraic expressions for the time derivatives of $S_0$ and $S_1$ respectively. Recall that $K_i:=k_{c_i}+k_{r_i}$. Throughout, we will view our set of polynomials $f$, which govern the ODE system, as a smooth function from $\RR^5$ to $\RR$. Any norm in this subsection will refer to the $\|\,\cdot\,\|_\infty$-norm on the respective space.

Define a domain $K^*$ such that $k_{f_i}\ll K_i$ for $i=1,2$ and such that it satisfies all hypotheses in the list at the start of this section (i.e., from  Proposition 2 of \cite{reduction}).

Note that for both substitutions, $E$ is bounded above by $E_{tot}\approx 0.65\mu E$. Hence, the norms on the slow variables are going to be approximately the same for both substitutions.

For the case (\ref{eq:e1}), we get that $\|f_\mathrm{red}\|$ (showing only the components that have changed compared to $f$) is given by
$$
\begin{Vmatrix}
k_{f_1}E(S_0-\det^{-1}\dot{S_0}(k_{f_2}S_1+K_2)+\dot{S_1}k_{f_2}S_0)-K_1C_1\\
k_{f_2}E(S_1-\det^{-1}(k_{f_1}(\dot{S_1}S_0+\dot{S_0}S_1))+\det^{-1}E(\dot{S_0}k_{f_1}k_{c_2}-\dot{S_1}k_{f_2})+k_{c_1}C_1-K_2C_2
\end{Vmatrix}
$$
on variables $C_1$ and $C_2$. Here, the quantity $\det$ is given by
$$\det=k_{f_1}(k_{c_1}+K_2)S_0+K_1(k_{f_2}S_1+K_2)=\mathcal{O}(K_1K_2).$$

Similarly, for (\ref{eq:e2}), we get
$$
\begin{Vmatrix}
k_{f_1}\left(ES_0-\frac{E+S_0}{ K_1}\dot{S_0}\right) - K_1C_1+\frac{k_{c_1}k_{f_2}}{ K_1K_2}(E+S_1)\dot{S_1}\\
k_{f_2}\left(ES_1-\frac{E+S_1}{ K_2}\dot{S_1}\right)+k_{c_1}C_1-K_2C_2
\end{Vmatrix}.
$$
In both cases, applying the triangle inequality 
shows that the upper entry is of order $|K_1C_1|$ while the lower entry is of order $|k_{c_1}C_1-K_2C_2|$. It follows that the quantity $R$, as in the assumptions of Proposition \ref{prop:full_prop}, is of a similar order for both substitutions.

For $\|Df_\mathrm{red}\|$, we get
$$
\begin{Vmatrix}
k_{f_1}E\left(1+\frac{k_{f_1}(\det +\dot{S_0}(k_{c_1}+K_2)}{\det^2}(k_{f_2}S_1+K_2)+k_{f_2}\dot{S_1}\right)\\
-k_{f_2}^2S_0 \\
-k_{f_1}\left(S_0-\frac{\dot{S_0}(k_{f_2}S_1+K_2)}{\det}+\dot{S_1}k_{f_2}S_0-E\frac{k_{r_1}(k_{f_2}S_1+K_2)}{\det}\right)-K_1 \\
0
\end{Vmatrix}^T
$$
for the row of the Jacobian corresponding to the entry of $f_\mathrm{red}$ giving the rate of change of $C_1$ and
$$
\begin{Vmatrix}
Ek_{f_1}\frac{k_{f_2}((k_{f_1}S_1-\dot{S_1})+(\dot{S_1}S_0+\dot{S_0}S_1)k_{f_1}(k_{c_1}+K_2))+k_{f_1}k_{c_2}\det+(\hat{S_1}k_{f_2}-k_{f_1}k_{c_2}\hat{S_0})(k_{c_1}+K_2)}{\det^2}\\
Ek_{f_2}\frac{\det^2+k_{f_1}k_{f_2}S_0\det+k_{f_1}k_{f_2}(\dot{S_1}S_0+\dot{S_0}S_1)K_1+k_{f_2}\det-(\dot{S_0}k_{f_1}k_{c_2}-\dot{S_1}k_{f_2})K_1}{\det^2}\\
-\frac{k_{f_2}k_{r_1}ES_1}{\det}+\frac{k_{f_1}k_{c_2}k_{r_1}E}{\det}+k_{c_1}\\
-\frac{k_{f_2}k_{r_2}ES_0}{\det}+E\frac{k_{f_2}k_{r_2}}{\det}+K_2
\end{Vmatrix}^T
$$
for the row of the Jacobian corresponding to the entry of $f_\mathrm{red}$ giving the rate of change of $C_2$ for the substitution (\ref{eq:e1}). Here, the four columns (or rows in the given transposed presentation) correspond to partial derivatives with respect to $S_0$, $S_1$, $C_1$ and $C_2$ (in that order). Note that the rows of $Df_\mathrm{red}$ corresponding to the rates of change of $S_0$, $S_1$, and $S_2$ are identical for both substitutions. Thus, we omit their calculations.

Similarly, for (\ref{eq:e2})
$$
\begin{Vmatrix}
k_{f_1}\left(S_0+E-\frac{2\dot{S_0}-k_{f_1}(E+S_0)}{ K_1}\right)+\frac{k_{c_1}k_{f_2}}{ K_1K_2}\dot{S_1} &
k_{f_1}\left(S_1-\frac{\dot{S_1}}{ K_2}\right) \\
k_{f_1}\left(S_0-\frac{\dot{S_0}}{ K_1}\right)-2k_{f_2}\frac{k_{c_1}k_{f_2}}{ K_1K_2} &
k_{f_1}\left(E+S_1-\frac{2\dot{S_1}-k_{f_1}(E+S_1)}{ K_2}\right) \\
-k_{f_1}\frac{\dot{S_0}+k_{r_1}(E+S_0)}{ K_1}-K_1 &
-k_{f_2}\frac{\dot{S_1}+k_{r_2}(E+S_1)}{ K_2}+k_{c_1} \\
0 &
K_2
\end{Vmatrix}^T.
$$

We note that the dominant terms are exactly $K_1$, $k_{c_1}$, and $K_2$ for both substitutions. We conclude that the quantity $L$, as in the assumptions of Proposition \ref{prop:full_prop}, is of similar a order for both substitutions.

For the quantity $\|f-f_\mathrm{red}\|$, for substitution (\ref{eq:e1}), we get
$$
\begin{Vmatrix}
k_{f_1}E(\dot{S_1}k_{f_2}S_0-\det^{-1}\dot{S_0}(k_{f_2}S_1+K_2))\\
\det^{-1}k_{f_1}k_{c_2}E\dot{S_0}-\det^{-1}k_{f_1}k_{f_2}E(\dot{S_0}S_1+\dot{S_1}S_0)-\det^{-1}k_{f_2}K_1E\dot{S_1}
\end{Vmatrix}.
$$

Similarly, for (\ref{eq:e2}) we get
$$
\begin{Vmatrix}
\frac{k_{c_1}k_{f_1}}{ K_1K_2}(E+S_1)\dot{S_1}-k_{f_1}\frac{E+S_0}{ K_1}\dot{S_0}\\
-k_{f_2}\frac{E+S_1}{ K_2}\dot{S_1}
\end{Vmatrix}.
$$
Both of the above norms are $\mathcal{O}(1)$ (as $\dot{S_0}$ and $\dot{S_1}$ contain terms $K_1$ and $K_2$ respectively).

Hence, all quantities use to bound accuracy in Proposition \ref{prop:full_prop} (above, by part (a), or below, by part (b)) are of a similar order. We conclude that the accuracy of both reductions is approximately the same when assuming $0\leq k_{f_i}\ll K_i$. This follows from the measurements $k_{M_i}\approx 25\,\mu M$ \cite{BarEven11}. Moreover, we have seen previously (in the main text) that $K_i>0$ for $i=1,2$ is a sufficient condition for a parameter value $\theta$ to be a QSS-parameter value in the sense of \cite{reduction}.

\subsection{Deriving Equality of ideals}\label{an:InitialConditions}

In this section of the appendix, we derive that $I(\mathcal{V}_0)=I_\Sigma$ for the ERK models with given initial condition.
In Saccomani et al. \cite{Saccomani2003}, $\mathcal{V}_0$ is defined as all trajectories of a given ODE system with initial condition $x_0$. That is, if $x(t)$ satisfies $S(0)=x_0$ and $\dot{S}=f(S, \theta)$, then 
$\mathcal{V}_0=\{x(t)\,\vert\,t\geq 0\}\subset\CC(\theta)^n$. Then $I(\mathcal{V}_0)$ is defined to be the ideal of all polynomials in $\CC(\theta)[S]$ vanishing on $\mathcal{V}_0$.

\begin{proposition}
Assume we are given a complex-valued ODE model, including an initial condition. Let $\theta$ denote the set of its parameters and $S$ the set of its variables. Consider the affine space $\mathcal{A}$ given by its variables and their derivative terms (of all orders), viewed as an affine space over the fraction field $\mathbb{C}(\theta)$. Define $\mathcal{V}_0$ to be the manifold in $\mathcal{A}$ given by the ODE trajectories under the given initial condition.

Then for the Full, Rational and Linear ODE model considered in this manuscript, we have that $I(\mathcal{V}_0)$, the differential ideal of all polynomials in the differential ring $\mathbb{C}(\theta)[S]$ vanishing on $\mathcal{V}_0$ (c.f. \cite{Saccomani2003}), is exactly $I_\Sigma$.
\end{proposition}

\begin{proof}
By definition, $I(\mathcal{V}_0)\supset I_\Sigma$. Hence, we need only show the inclusion $I_\Sigma \subset I(\mathcal{V}_0)$. by way of contradiction, suppose there are elements of $I(\mathcal{V}_0)$ which are not elements of $I_\Sigma$. Without loss of generality, we may assume that such polynomials contain no derivatives of variables. Indeed, starting with any polynomial in $I(\mathcal{V}_0)$, we can use the differential equations in $I_\Sigma$ to cancel out all terms containing derivatives of variables and so obtain an element of $I(\mathcal{V}_0)$ without such terms.

We need to take some extra care for the Rational ERK model: strictly speaking, the polynomial $A(\gamma_1S_0+\gamma_2S_1+1)-1$ is in $I(\mathcal{V}_0)$ but not in $I_\Sigma$. However, we know this relation \textit{a priori} and it is merely a result of us converting a rational ODE model into a polynomial one. In particular, it is independent of an initial condition. Hence, in the context of this section, we will assume that the Rational ERK model contains an output variable $y_3:=A(\gamma_1S_0+\gamma_2S_1+1)$ (in addition to $y_1=S_0$ and $y_2=S_1$). Whilst $y_3$ adds polynomials to $I_\Sigma$, we will omit it from any polynomials we study henceforth as it is equivalent to the constant 1 (given the definition of $A$).

For the Rational ERK model, assume there exists a non-zero $p\in\mathbb{C}(\theta)[S_0, S_1, A]$ such that $p(S_0, S_1, A)=0$ for all $t\geq0$ given our initial condition. By using $A(\gamma_1S_0+\gamma_2S_1+1)=1$, we can turn $p$ into a rational function in variables $S_0, S_1$ which vanishes on the ODE trajectories. Hence, there must exist a polynomial $q\in\mathbb{C}(\theta)[S_0, S_1]$ such that $q(S_0, S_1)=0$ for all $t\geq0$ given our initial condition. Any properties we will use in the following argument are satisfied by both the Rational and Linear ERK models unless stated otherwise. As we need to disprove the existence of such polynomial $q$ when studying the Linear ERK model the following conclusion will hold for both the Linear and Rational ERK models.

Without loss of generality, assume that $q$ is a polynomial of the smallest non-zero degree satisfying $q(S_0, S_1)=0$ on the ODE trajectories and that $q$ is irreducible. If it were reducible, simply replace $q$ with an irreducible factor of $q$ such that $q(S_0,S_1)=0$ for all $0\leq t<t_1$ for some $t_1>0$. The following argument will still hold.

The Linear and Rational ERK models have thus far been presented as an ODE in $\RR^2$ in the basis $S_0$, $S_1$. At generic parameter values (i.e.,if $\kappa_1\neq\kappa_2$), we may change our basis to $X_0=S_1$, $X_1=\kappa_1(1-\pi)/(\kappa_1-\kappa_2)S_0+S_1$. We then find $\dot{X_0}=-\kappa_1X_0$ and $\dot{X_1}=-\kappa_2X_1$ for the Linear ERK model and $\dot{X_0}=-\kappa_1AX_0$ and $\dot{X_1}=-\kappa_2AX_1$ for the Rational ERK model.

As this change of basis is an invertible affine transformation, we may assume that $q$ is a polynomial in variables $X_0$ and $X_1$. We note that this diagonalisation of the ODE models would not have been possible if we had not removed $S_2$ from our systems, as $S_2$ decouples.

We will write
$$q(X_0, X_1)=\sum_{i,j\geq0}a_{ij}X_0^iX_1^j=0.$$
Taking the derivative of this equation with respect to $t$ gives
$$\sum_{i,j\geq0}a_{ij}(i\dot{X_0}X_0^{i-1}X_1^j+j\dot{X_1}X_0^iX_1^{j-1})=0,$$
which, after substitution for the derivative variables, yields 
$$q'(X_0,X_1):=\sum_{i,j\geq0}a_{ij}(i\kappa_1+j\kappa_2)X_0^iX_1^j=0$$
(for the Rational ERK model, we need to divide by $A$ to obtain such $q'$).

Note that, given our initial condition, both $X_0$ and $X_1$ vary, hence the intersection of $V(q)$ and $V(q')$ must contain a smooth point, as both contain the ODE trajectory. This would imply that $q'$ is a constant multiple of $q$, which is not true at generic parameter values. Hence, no such $q$ can exist.

For the Full ERK model, assume that there is a non-zero polynomial $p\in\mathbb{C}(\theta)[S_0, S_1, C_1, C_2]$ such that $p(S_0, S_1, C_1, C_2)=0$ for all $t\geq0$ given our initial condition. Let $p$ have degree $n$. Again, we may assume without loss of generality that $n$ is minimal. We know that there are no conserved quantities in the Full ERK model (after removing $E$ and $S_2$) and thus $n\geq2$. Let
$$p(S_0, S_1,C_1,C_2)=\sum_{i,j,k,l\geq0}a_{ijkl}S_0^iS_1^jC_1^kC_2^l=0.$$
By taking the derivative with respect to $t$, we get
$$\sum_{i,j,k,l\geq0}a_{ijkl}(i\dot{S_0}S_0^{i-1}S_1^jC_1^kC_2^l+j\dot{S_1}S_0^iS_1^{j-1}C_1^kC_2^l+k\dot{C_1}S_0^iS_1^{j}C_1^{k-1}C_2^l+l\dot{C_2}S_0^iS_1^{j}C_1^kC_2^{l-1})=0.$$

For $S_0$, this rearranges to
$$\dot{S_0}=-\frac{\sum_{i,j,k,l\geq0}a_{ijkl}(j\dot{S_1}S_0^iS_1^{j-1}C_1^kC_2^l+k\dot{C_1}S_0^iS_1^{j}C_1^{k-1}C_2^{l}+l\dot{C_2}S_0^iS_1^{j}C_1^kC_2^{l-1})}{\sum_{i,j,k,l\geq0}a_{ijkl}iS_0^{i-1}S_1^{j}C_1^{k}C_2^l},$$
(assuming there are $S_0$ terms in $p$) and we can derive similar expressions for the other variables. Note that the denominator is of a smaller degree than $p$, hence it will be non-zero at almost all points of our ODE trajectory. In addition, for our given initial condition, all of our variables vary and hence the numerator must be a non-zero polynomial and must not vanish on the ODE trajectory. Denote $i^*$ the highest power of $S_0$ in $p$. Then $a_{i^*jkl}\neq0$ implies $k=0$, as otherwise the highest power of $S_0$ in the numerator is $i^*+1$ and $i^*-1$ in the denominator, implying that we can find an $S_0^2$-term in $\dot{S_0}$, a contradiction. We can derive similar statements for $S_1$, $C_1$, and $C_2$. Then, in the above equation, the highest power of $C_1$ in the numerator is $k^*+1$, while the highest power in the denominator is $k^*-1$ (at generic parameter values), implying that $\dot{S_0}$ is quadratic in $C_1$, a contradiction. If $p$ does not contain terms in $S_0$, we can apply a similar argument to $S_1$, $C_1$ or $C_2$.

In conclusion, no such $p$ can exist.
\end{proof}

\subsection{The Linear ERK model is globally structurally identifiable}\label{ap:LinearSIdirect}
\begin{proposition}
For any choice of three time points $0<t_1<t_2<t_3$ the model prediction map $\phi_{t_1,t_2,t_3}$ is injective and so the Linear ERK model is globally structurally identifiable.
\end{proposition}

\begin{proof}
For a parameter $(\kappa_1,\kappa_2,\pi)$, we denote the analytic solution in the Linear ERK model for species $i$ at time $t$ by ${S}_i(t)$ as in Subsection \ref{subsec:PILinear}. Then the model prediction map $\phi_{t_1,t_2,t_3}$ is given by
\[(\kappa_1,\kappa_2,\pi)\mapsto ({S}_0(t_1),{S}_1(t_1),{S}_2(t_1),{S}_0(t_2),{S}_1(t_2),{S}_2(t_2),{S}_0(t_3),{S}_1(t_3),{S}_2(t_3))
\]

Suppose that $(\kappa_1,\kappa_2,\pi)$ and $(\kappa_1',\kappa_2',\pi')$ are two parameters such that 
$$\phi_{t_1,t_2,t_3}(\kappa_1,\kappa_2,\pi)=\phi_{t_1,t_2,t_3}(\kappa_1',\kappa_2',\pi').$$ 
Looking at the first component we find that 
$5e^{-\kappa_1t_1}=5e^{-\kappa_1't_1}$, and since $t_1\neq0$ it follows that $\kappa_1=\kappa_1'$.

There are three cases to consider. The first case is if $\kappa_2=\kappa_1=\kappa_2'$. In this case, looking at the second component, we get that $5\kappa_1(1-\pi)t_1e^{-\kappa_1t_1}=5\kappa_1(1-\pi')t_1e^{-\kappa_1t_1}$, and so $\pi=\pi'$, since $t_1\neq0$. 

Without loss of generality, the second case is
if $\kappa_2=\kappa_1\neq\kappa_2'$. Suppose for a contradiction that $(\kappa_1,\kappa_1,\pi)\neq(\kappa_1,\kappa_2',\pi')$. Let ${S}_1(t)$ be the analytic solution of the Linear ERK model for species 1 at time $t$ with parameter $(\kappa_1, \kappa_1,\pi)$ and ${S}'_1(t)$ be the analytic solution at time $t$ with parameter $(\kappa_1,\kappa'_2,\pi')$.
For the moment, consider $t>0$ to be a variable. Then ${S}_1(t)={S'}_1(t)$ is equivalent to
\begin{equation}\label{eq:LinSIcase2}
(1-\pi)te^{-\kappa_1t}=\frac{1-\pi'}{\kappa_1-\kappa_2'}\left(e^{-\kappa'_2t}-e^{-\kappa_1t}\right).
\end{equation}
Dividing both sides by $e^{-\kappa_1t}$ and rearranging the above gives
$$(1-\pi)t-\frac{1-\pi'}{\kappa_1-\kappa_2'}e^{-(\kappa_2'-\kappa_1)t}=-\frac{1-\pi'}{\kappa_1-\kappa_2'} $$
Taking the derivative with respect to $t$ yields
$$(1-\pi)=(1-\pi') e^{-(\kappa_2'-\kappa_1)t}.$$
Rearranging and taking a $\log$ then gives
$$t=\frac{1}{\kappa_1-\kappa_2'}\log \left(\frac{1-\pi}{1-\pi'}\right) .$$
As the derivative of Equation (\ref{eq:LinSIcase2}) has exactly one solution in $t$, Equation (\ref{eq:LinSIcase2}) has at most two solutions in $t$ by Rolle's theorem. Therefore, the second, fifth and eighth components of $\phi_{t_1,t_2,t_3}$ cannot take the same value. This is a contradiction, and so we should have $(\kappa_1,\kappa_1,\pi)=(\kappa_1,\kappa_2',\pi')$, meaning that this case is simply not possible.

We now consider the third and final case, $\kappa_2\neq\kappa_1\neq\kappa_2'$. Suppose for a contradiction that $(\kappa_1,\kappa_2,\pi)\neq(\kappa_1,\kappa_2',\pi')$. Let ${S}_1(t)$ be the analytic solution of the Linear ERK model for species 1 at time $t$ with parameter $(\kappa_1, \kappa_2,\pi)$ and ${S'}_1(t)$ be the analytic solution at time $t$ with parameter $(\kappa_1,\kappa'_2,\pi')$. 
For the moment, consider $t>0$ to be a variable. Then ${S}_1(t)={S}'_1(t)$ is equivalent to
\begin{equation}
\frac{1-\pi}{\kappa_1-\kappa_2}\left(e^{-\kappa_2t}-e^{-\kappa_1t}\right)=\frac{1-\pi'}{\kappa_1-\kappa_2'}\left(e^{-\kappa'_2t}-e^{-\kappa_1t}\right).\label{eq:app}
\end{equation}
Dividing both sides by $e^{-\kappa_1t}$ and rearranging the above gives
$$\frac{1-\pi}{\kappa_1-\kappa_2}e^{-(\kappa_2-\kappa_1)t}-\frac{1-\pi'}{\kappa_1-\kappa'_2}e^{-(\kappa'_2-\kappa_1)t}=\frac{1-\pi}{\kappa_1-\kappa_2}-\frac{1-\pi'}{\kappa_1-\kappa'_2}.$$
Taking the derivative with respect to $t$ yields
$$(1-\pi)e^{-(\kappa_2-\kappa_1)t}-(1-\pi')e^{-(\kappa'_2-\kappa_1)t}=0.$$
Taking a $\log$ and rearranging gives
$$t=\frac{1}{\kappa_1-\kappa_2'}\log \left(\frac{1-\pi}{1-\pi'}\right) .$$

As the derivative of Equation (\ref{eq:app}) has exactly one solution in $t$, Equation (\ref{eq:app}) has at most two solutions in $t$ by Rolle's theorem. Therefore, the second, fifth and eighth components of $\phi_{t_1,t_2,t_3}$ cannot take the same value. This is a contradiction, and so we must have $(\kappa_1,\kappa_2,\pi)=(\kappa_1,\kappa_2',\pi')$.

\end{proof}

\end{document}